\documentclass{article}

\PassOptionsToPackage{numbers, compress}{natbib}
\usepackage[numbers,square]{natbib}

\usepackage{PRIMEarxiv}

\usepackage[utf8]{inputenc} % allow utf-8 input
\usepackage[T1]{fontenc}    % use 8-bit T1 fonts
\usepackage{hyperref}       % hyperlinks
\usepackage{url}            % simple URL typesetting
\usepackage{booktabs}       % professional-quality tables
\usepackage{amsfonts}       % blackboard math symbols
\usepackage{nicefrac}       % compact symbols for 1/2, etc.
\usepackage{microtype}      % microtypography
\usepackage{xcolor}         % colors
\usepackage{microtype}
\usepackage{graphicx}
\usepackage{subcaption}
\usepackage{enumitem}
\usepackage{algorithm}
\usepackage{algorithmic}
\usepackage{booktabs}
\usepackage{amsmath}
\usepackage{float}
\usepackage{placeins}
\usepackage{tikz}
\usepackage{threeparttable}
\usetikzlibrary{shapes, arrows.meta, positioning, fit, backgrounds}
\usepackage{booktabs}   % For professional looking rules (lines)
\usepackage{multirow}   % For multi-row table cells
\usepackage[table]{xcolor} % For row coloring
\usepackage{caption}    % For better caption spacing
\usepackage{makecell}
\renewcommand{\algorithmiccomment}[1]{// \texttt{\footnotesize\raggedright #1}}

\newcommand{\PTO}[1]{PTO}
\newcommand{\Acc}[1]{accuracy-based loss function}

\usepackage{amsmath,amsfonts,bm}
\usepackage{enumitem}
\def\eqref#1{equation~\ref{#1}}
\def\1{\bm{1}}

\DeclareMathAlphabet{\mathsfit}{\encodingdefault}{\sfdefault}{m}{sl}
\SetMathAlphabet{\mathsfit}{bold}{\encodingdefault}{\sfdefault}{bx}{n}

\DeclareMathOperator*{\argmax}{arg\,max}

\usepackage{url}

\usepackage{hyperref}

\usepackage{amsmath}
\usepackage{amssymb}
\usepackage{mathtools}
\usepackage{amsthm}

\usepackage[capitalize,noabbrev]{cleveref}

\newenvironment{theorem*}[2][]
  {\par\noindent\textbf{Restatement of Theorem #2\if\relax#1\relax\else\ (#1)\fi.}\itshape}
  {\par}
\newenvironment{definition*}[2][]
  {\par\noindent\textbf{Restatement of Definition #2\if\relax#1\relax\else\ (#1)\fi.}\itshape}
  {\par}
\newenvironment{assumption*}[2][]
  {\par\noindent\textbf{Restatement of Assumption #2\if\relax#1\relax\else\ (#1)\fi.}}
  {\par}
\newenvironment{proposition*}[2][]
  {\par\noindent\textbf{Restatement of Proposition #2\if\relax#1\relax\else\ (#1)\fi.}\itshape}
  {\par}

\newenvironment{corollary*}[2][]
  {\par\noindent\textbf{Restatement of Corollary #2\if\relax#1\relax\else\ (#1)\fi.}\itshape}
  {\par}
  
 \theoremstyle{plain}
\newtheorem{theorem}{Theorem}[section]
\newtheorem{proposition}[theorem]{Proposition}

\newtheorem{corollary}[theorem]{Corollary}
\newtheorem{claim}[theorem]{Claim}
\theoremstyle{definition}
\newtheorem{definition}[theorem]{Definition}
\newtheorem{assumption}[theorem]{Assumption}
\theoremstyle{remark}
\newtheorem{remark}[theorem]{Remark}

\usepackage[textsize=tiny]{todonotes}

\title{Strategic Decision Focused Learning}
\author{%
  Tinashe Handina\thanks{Corresponding author: \texttt{thandina@caltech.edu}.} \\
  California Institute of Technology\\
  \And
  Yuehan Diao \\
  University of Chicago \\
  \AND
  Adam Wierman \\
  California Institute of Technology \\
  \And
  Eric Mazumdar \\
  California Institute of Technology \\
}

\begin{document}

\maketitle

\begin{abstract}
Machine learning (ML) predictions are increasingly being used to guide decision-making, giving rise to the problem of decision-focused learning (DFL) where predictors are optimized for downstream decision quality rather than accuracy alone. However, most existing work assumes a single decision-maker optimizing in isolation. This paper formalizes \emph{strategic} decision-focused learning, where an ML system predicts an exogenous state that some agents observe before playing a game. For example, a park ranger may predict wildlife locations to allocate anti-poaching patrols against strategic poachers. While the exogenous state is unaffected by agent actions, predictions influence agents' strategies and the resulting equilibrium. We find that strategic considerations fundamentally change the learning problem. In particular, we show the prediction accuracy–equilibrium payoff landscape can be non-monotonic---i.e., better predictions can degrade performance. We propose algorithmic approaches to address these challenges and validate them across benchmarks in wildlife conservation and infrastructure protection. Our theory and experiments highlight the importance of accounting for strategic interactions when designing predictors.
\end{abstract}

%pre The realization of the full potential of ML, however, will require an adaptation of conventional paradigms to fit the particularities of strategic interactions. This paper provides a framework for utilizing ML predictions in decision-making in strategic settings. We argue for an end-to-end approach where the impacts of a prediction on an agent's utility are explicitly taken into account during training as opposed to just optimizing the accuracy of a prediction. We motivate this approach theoretically by illustrating 

\section{Introduction}
\label{sec: introduction}

Machine Learning (ML) systems are playing an increasingly vital role in shaping decisions across a wide range of domains. From commerce \cite{commerceML}, to security \cite{krever2025guard} \cite{tang2024hierarchical}, to government \cite{mhasawade2021machine} \cite{zhang2024socialenvironment}, decision makers are increasingly relying on the outputs of ML systems as the basis upon which they optimize their actions. Driving the proliferation of ML is its demonstrable ability to understand and predict the evolution of complex systems reasonably well. As ML systems become an ever more consequential tool used for decision making, understanding when standard training objectives are \emph{structurally inadequate} — not merely suboptimal — is of vital importance. This paper identifies one such setting.

Learning with the view that the output of a learning algorithm is going to be used in a downstream optimization problem is an active area of research often called decision-focused learning\cite{mandi2024decision}\cite{wilder2019melding}\cite{shah2022decision}. Research in this area, to date, has often focused on a single decision maker leveraging ML outputs to optimize a utility or cost function in isolation \cite{shah2022decision}. In many consequential applications, however, decision makers find that they are interacting with other actors who often act in a \emph{strategic} manner. This introduces complexities to the optimization problem as the agent now not only has to consider the impact of the decision on the cost function, but also has to take into account how other actors may respond to any particular decision they make.
\begin{figure}[t]
\centering
\includegraphics[width=0.85\linewidth]{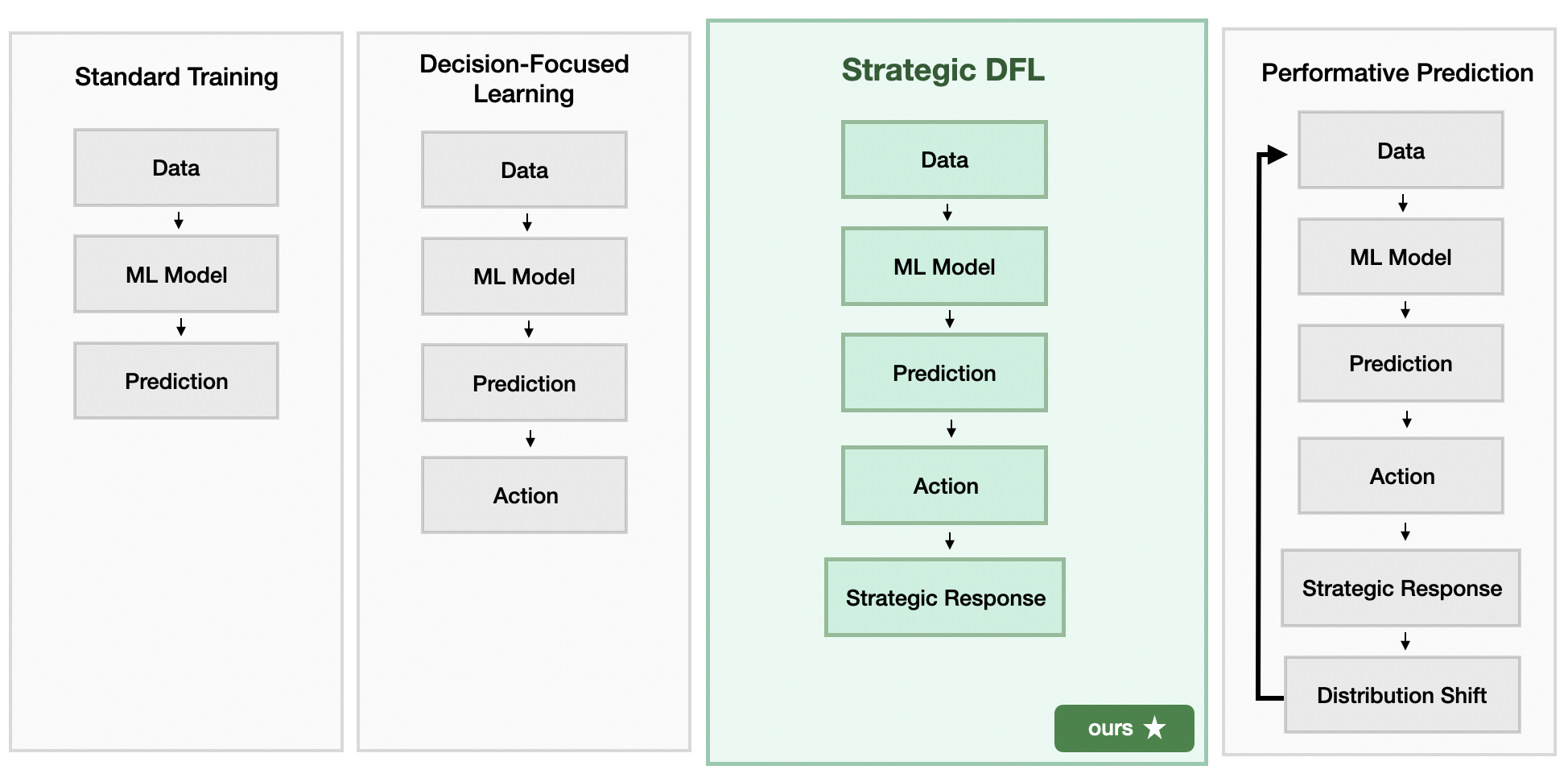}
\caption{Comparison of different learning paradigms. In contrast to classical and decision-focused learning, SDFL operates in a strategic context, significantly complicating the learning task. Unlike performative prediction or strategic classification, the data distribution is exogenous—free from strategic feedback loops.
% classical learning and decision-focused learning, the data-driven prediction learned in our proposed framework of SDFL is used in a strategic context---significantly complicating the task of learning. In contrast to other paradigms that seek to incorporate the downstream effects of strategic decision-making into the learning framework---e.g., performative prediction or strategic classification---the distribution of data is exogenous and is not subject to strategic feedback loops.
}
\label{fig:overview}
\end{figure}

In this work, we formalize this problem area which we call \textbf{strategic decision-focused learning} (SDFL). In this setting, an agent is trying to leverage ML systems to aid with the prediction of some \emph{exogenous} state that affects a game that they play. Such settings are prevalent across a wide range of domains, including:
\begin{itemize}
\item \textbf{Electricity markets:} ML predictors can be used by producers competing in grid markets to forecast exogenous renewable generation conditions (e.g., wind and solar availability), which are then used to inform pricing strategies that interact with those of competing producers.

\item \textbf{Motorsport strategy:} In Formula One, teams leverage predictive models to forecast exogenous variables such as tire degradation rates and safety car probabilities, which shape pit-stop and fuel strategies in direct competition with opposing teams.
\item \textbf{Security and conservation:} Predicting population distributions over a geographic region — such as wildlife density or civilian movement — informs patrol strategies in adversarial security games (we evaluate this concretely in conservation and infrastructure protection domains in this work).
\end{itemize}
We find that the explicit characterization of the downstream strategic interaction leads to fundamental — not merely practical — considerations for the design and use of machine learning systems. To that end, this work seeks to address the following question:
\begin{quote}
\centering
{\em How should agents develop and make use of Machine Learning for exogenous state predictions in strategic interactions?}
\end{quote}
Traditional approaches to this problem in the strategic domain have often relied on first predicting exogenous states such as demand \cite{tarallo2019machinedemand} or population distributions \cite{robson2017savanna}, then optimizing for a decision. We show that this approach is not merely suboptimal but also suffers from a fundamental issue of misalignment. As we show, in strategic settings no monotone transformation of a supervised accuracy loss can serve as a reliable proxy for strategic utility. This failure is structural, arising from the interaction between prediction error and equilibrium response, and persists regardless of model capacity or data quantity.

\textbf{Our contributions:}
This paper makes three conceptual contributions and one algorithmic contribution.

\textbf{(C1) Problem formalization.} We formalize strategic decision-focused learning as a distinct problem within the decision-focused learning paradigm. The SDFL problem is defined for general games and solution concepts: an agent predicts an exogenous state and commits to an equilibrium strategy based on that prediction, while opponents respond under the \emph{true} state. This split between the predicted world the agent plans for and the true world opponents respond to is the source of all the difficulty we characterize. The Stackelberg equilibrium is one tractable instantiation we study in depth; the formalization and hardness results hold for general solution concepts.

\textbf{(C2) Structural hardness.} We prove that the mapping from prediction accuracy to equilibrium payoff is \emph{provably non-monotone}: a less accurate prediction can strictly outperform a more accurate one (Proposition~3.2), while the true state in, some instances, may be suboptimal (Proposition~3.3). Together, these results characterize the structure of the accuracy--utility landscape. This is not a worst-case pathology: it is a feature of the landscape, implying that standard predict-then-optimize pipelines are not well motivated for SDFL.

\textbf{(C3) Geometric characterization of hardness.} We prove that, for a class of games, as game dimensionality grows, an increasing proportion of strategies lie near best-response boundaries (Theorem~3.7), explaining geometrically \emph{why} small prediction errors may induce large utility swings. This result characterizes why the hardness can intensify as the action space grows.

\textbf{(C4) Initial algorithmic framework.} Motivated by C1--C3, we propose end-to-end prediction approaches as an initial algorithmic response to SDFL. We present these as proof-of-concept rather than definitive solutions, in the same spirit as early task-loss surrogates in single-agent DFL that preceded formal analysis \cite{donti2017task, wilder2019melding}. Experiments across two real-world security game domains validate that end-to-end approaches outperform predict-then-optimize despite higher supervised prediction error, corroborating our theoretical findings that accuracy and strategic utility are misaligned.

\subsection{Related work}
 Learning how to make use of machine learning for the purposes of downstream decision making has emerged as an important area of research~\cite{mandi2024decision, elmachtoub2022smart}. The key idea powering much of this work is that embedding the optimization problem of interest into the training loop and backpropagating decision-relevant gradients can lead to better performance than first optimizing a prediction and then optimizing for a decision~\cite{donti2017task, Yeh2024EndtoEndCC}. While standard approaches to decision-focused learning often focus on a single agent optimizing in isolation~\cite{shah2022decision}, we focus on extending and formalizing this paradigm for the strategic domain. In particular, we consider what happens when the decision is observed by strategic agents whose choices are affected by the optimization decision itself. We call this setting \emph{strategic decision-focused learning}. Crucially, the strategic setting is not merely a harder instance of standard DFL but a qualitatively different problem: the accuracy--utility misalignment we identify has no analogue in the single-agent case. Closest to our work is game-focused learning in security games~\cite{perrault2020endtoend, wang2020scalable}: there the payoffs are known, and a defender learns how an attacker reacts and then best-responds to what is learned. In our framework, the state that sets the payoffs, which our opponents already observe, is what is predicted. The predictor then plays an equilibrium of the game it thinks it is in. This introduces new forms of complexity: for example, predicting the true state in some instances may be suboptimal (Proposition~3.3).

\paragraph{Learning in strategic environments:} Recent years have seen a surge of research on deploying learning algorithms within strategic environments \cite{mert, mazumdar}. This focus on the intersection of learning and strategic interactions has given rise to, among others, the domains of strategic classification~\cite{hardt16strategic}, adversarial machine learning~\cite{madry_adversarial}, and multi-agent reinforcement learning~\cite{zhang2021multi}. Closely related is performative prediction~\cite{perdomo}, wherein the act of prediction itself causes a distributional shift in the data the model is trained on. Our work takes an orthogonal track: rather than the prediction changing the input distribution, we ask how to properly account for exogenous state information within a strategic decision-making task where the prediction shapes the equilibrium that emerges. The exogeneity of the data distribution is what separates SDFL from performative prediction; the presence of strategic opponents is what separates it from standard DFL. A key theme arising from this body of work is that learning in strategic environments is often non-trivial, with phenomena such as instability of algorithms resulting in convergence to cycles and chaos~\cite{cycles}, and strategic manipulations leading to improved causal discovery~\cite{gaminghelps}. Of particular interest is the insight that small changes in learning parameters may lead to different equilibria with different utilities~\cite{zrinc21follow}. Previous work has also found that model selection within the broader context of learning in strategic environments is complicated, with the model expressivity--payoff landscape in general being non-monotonic~\cite{handina2024understanding}. Our work contributes a precise characterization of this non-monotonicity in the SDFL setting, grounding it in the geometry of best-response regions rather than model expressivity.

\section{Problem Formulation}
\label{sec: preliminaries}

We formalize strategic decision-focused learning as a distinct problem within the 
decision-focused learning paradigm. We begin with a general formulation applicable 
to any game and solution concept, then instantiate it for the Stackelberg setting 
we study in depth.

In the standard decision-focused learning (DFL) setting, a single agent uses a 
learned model $f_\theta : \mathcal{X} \to \mathcal{S}$ to predict an uncertain 
parameter $s \in \mathcal{S}$ and solves a downstream optimization problem whose 
objective depends on $s$. The central insight of DFL is that training $f_\theta$ 
to minimize a supervised loss can be suboptimal when downstream decision quality 
is what matters~\cite{elmachtoub2022smart, wilder2019melding}. Critically, most existing 
DFL assumes the agent optimizes \emph{in isolation} --- no other agent observes or 
responds to the decision. As we show, this isolation assumption is not merely a 
simplification but a load-bearing one: its violation changes the learning problem 
qualitatively.

\subsection{Strategic Decision-Focused Learning}

We consider a game $G = (\mathcal{N}, \{\mathcal{A}_i\}_{i \in \mathcal{N}}, 
\{U_i\}_{i \in \mathcal{N}}, \mathcal{S})$ where $\mathcal{N}$ is a finite set 
of agents, $\mathcal{A}_i$ is the action set of agent $i$, $\mathcal{S} \subseteq 
\mathbb{R}^d$ is an exogenous state space, and $U_i : \prod_{j} \mathcal{A}_j 
\times \mathcal{S} \to \mathbb{R}$ is agent $i$'s utility. Each agent selects a 
mixed strategy $\alpha_i \in \Delta(\mathcal{A}_i)$.

The state $s \in \mathcal{S}$ is \emph{exogenous}: drawn from a fixed distribution 
$P_\mathcal{S}$ independently of any agent's actions. A primary agent $i$ observes 
only a dataset $\mathcal{D} = \{(x^{(t)}, s^{(t)})\}_{t=1}^T$ and learns a 
predictor $f_\theta : \mathcal{X} \to \mathcal{S}$. All other agents observe the 
true realized state $s_r$.

Let $\sigma : \mathcal{S} \to \prod_{j} \Delta(\mathcal{A}_j)$ denote a 
\emph{solution mapping} returning an equilibrium profile for a given state. Agent 
$i$ commits to $\sigma_i(\hat{s})$ based on their prediction $\hat{s} = f_\theta(x)$, 
while opponents respond under the true state $s_r$. The realized joint profile is:
\(\sigma(\hat{s}, s_r) = \big(\sigma_i(\hat{s}),\; \mathrm{Eq}_{-i}(\sigma_i(\hat{s}),\; s_r)\big)\) where $\mathrm{Eq}_{-i}(\alpha_i, s)$ is the equilibrium response of agents $-i$ 
to a fixed $\alpha_i$ under state $s$. The \emph{strategic decision-focused 
learning problem} is then:
\begin{equation}
    \theta^* = \arg\max_\theta \; \mathbb{E}_{(x, s_r) \sim \mathcal{D}}\left[ 
    U_i\!\left(\sigma(f_\theta(x),\, s_r),\; s_r\right) \right]
    \label{eq:sdfl_objective}
\end{equation}

The split between the predicted world agent $i$ plans for and the true world 
opponents respond to is the source of all the difficulty we characterize in 
Section~\ref{sec: limits}. We note that this formulation is distinct from 
performative prediction~\cite{perdomo}, wherein predictions shift $P_\mathcal{S}$ 
through a feedback loop --- here the data distribution is fixed and predictions 
affect only the equilibrium. Our framework is clearly distinguished from existing learning problems in Figure~\ref{fig:overview}.

\subsection{Stackelberg Instantiation}

The SDFL objective~\eqref{eq:sdfl_objective} is defined for general games and 
solution concepts; the hardness results of Section~\ref{sec: limits} hold in 
this generality. For the algorithmic contributions and experiments of 
Sections~\ref{sec: e2e}, we instantiate SDFL under 
the \emph{Stackelberg solution concept}, chosen for tractability and its natural 
fit to security applications.

We focus on two-player interactions in which agent $i$ leads and agent $-i$ 
follows. Given predicted state $\hat{s}$, the leader computes:
\( \bar{\alpha}_i(\hat{s}) = \arg\max_{\alpha_i \in \Delta(\mathcal{A}_i)} 
    U_i\!\left(\alpha_i,\; BR_{-i}(\alpha_i, \hat{s}),\; \hat{s}\right)\)
where $BR_{-i}(\alpha_i, s) = \arg\max_{\alpha_{-i}} U_{-i}(\alpha_i, \alpha_{-i}, 
s)$. Since the follower responds under the true state $s_r$, the leader's realized 
utility is:
\(U_i^{\mathrm{realized}}(\hat{s}, s_r) = U_i\!\left(\bar{\alpha}_i(\hat{s}),\; 
    BR_{-i}(\bar{\alpha}_i(\hat{s}),\; s_r),\; s_r\right)\).
Stackelberg equilibria are the natural solution concept for security settings, 
where attackers are assumed to observe both the state and the defender's committed 
strategy~\cite{yang2014adaptive}.

\paragraph{Illustrative example.} A game ranger predicts elephant locations to 
allocate patrols against strategic poachers. The exogenous state --- the true 
elephant distribution --- is unknown to the ranger but (in the worst case) observed by the poachers. 
The ranger commits to a patrol strategy based on their prediction; the poachers 
best-respond knowing the truth. The ranger's realized utility depends not on how 
accurately they predicted elephant locations, but on how well their patrol strategy 
performs against opponents who know the true state. This gap between prediction 
quality and strategic performance is precisely what~\eqref{eq:sdfl_objective} 
captures.

\section{Limitations of Accuracy based Predictions in Games}
\label{sec: limits}
We establish that standard accuracy-based approaches to SDFL suffer from a problem of mis-alignment. We prove two independent results that together characterize why: first, 
the mapping from prediction accuracy to equilibrium payoff is provably 
non-monotone, implying no monotone transformation of a supervised loss 
can serve as a reliable proxy for strategic utility; second, for a broad 
class of games, the fraction of strategies near best-response boundaries 
grows with game dimensionality, explaining geometrically why small 
prediction errors induce large utility swings. While we
instantiate these results in the Stackelberg setting for concreteness, 
the non-monotonicity result holds for the general SDFL 
objective~\eqref{eq:sdfl_objective} under other solution mappings, $\sigma$, 
satisfying the conditions of Section~\ref{sec: preliminaries}.

Given a dataset $\mathcal{D}$ relating to the exogenous state $s$, standard approaches to learning in strategic environments first find some $s_{pred}$ by minimizing a statistical loss function $\mathcal{L}(\phi, \mathcal{D})$ which is proportional to some accuracy metric $\phi$ for the game instance. After having learned a prediction of the state, the agent then looks to optimize their decision-making given their state prediction, which was yielded from the previous estimation step.
\begin{definition}[Accuracy Metric]
    \label{def:accuracy_metric}
    Let $\mathcal{S} \subseteq \mathbb{R}^n$ for some $n \in \mathbb{N}$, and let $s_r \in \mathcal{S}$ be the true realization of the state. An accuracy metric for the game instance is a function $\phi: \mathcal{S} \rightarrow \mathbb{R}_{\geq 0}$ defined by $\phi(s) = \|s - s_r\|$ for some symmetric norm $\|\cdot\|$ on $\mathbb{R}^n$.
\end{definition}

To extend the previous example of elephant conservation, we assume that the ranger is given some dataset $\mathcal{D}$ on the location of elephants within the reserve. Standard approaches to the strategic decision-focused learning problem would have the ranger find $s_{pred}$ (a prediction of where the elephants would be), and then from there, they would optimize to find a patrol strategy that maximizes their utility. We call this approach to learning \textbf{``Predict Then Optimize''} (PTO).

We begin by showing the limitations of this particular learning paradigm. We find that it is possible for a state prediction to be less accurate yet lead the decision maker into selecting an action that yields better payoff when compared to some other, more accurate state prediction. This shows that utility within the context of strategic environments can be non-monotonic with respect to prediction accuracy.  

\begin{proposition}
There exists a game instance with true state $s_r$, and predictions $s_{p,1}$ and $s_{p,2}$, wherein $s_{p,1}$ is strictly more accurate than $s_{p,2}$ under every accuracy metric $\phi$, yet player $i$'s utility at the equilibrium induced by $s_{p,1}$ is strictly lower than that induced by $s_{p,2}$:
\[
U_i\!\left(\alpha^*(s_{p,1}),\, s_r\right)
\;<\;
U_i\!\left(\alpha^*(s_{p,2}),\, s_r\right),
\]
where $\alpha^*(s) = (\alpha_i^*(s),\, BR_{-i}(\alpha_i^*(s), s_r))$ denotes the equilibrium joint strategy profile when the primary plays assuming state $s$, the rest of the agents see the true state $s_r$, and utility is evaluated at $s_r$.

    \label{prop:non-monotone}
\end{proposition}

In addition, we find that there exist game instances in which the true state $s_r$ is not the utility-optimal prediction for a player: a strategically calibrated misprediction can induce an equilibrium that, when played against the true state, yields strictly higher utility than the equilibrium induced by $s_r$ itself.

\begin{proposition}
There exists a game instance with true state $s_r$ and player $i$ such that there exists a prediction $s_p \in \mathcal{S}$ satisfying:
\[
U_i\!\left(\alpha^*(s_p),\, s_r\right) \;>\; U_i\!\left(\alpha^*(s_r),\, s_r\right),
\]
where $\alpha^*(s) = (\alpha_i^*(s),\, BR_{-i}(\alpha_i^*(s), s_r))$ denotes the equilibrium joint strategy profile when the primary plays assuming state $s$, the rest of the agents see the true state $s_r$, and utility is evaluated at $s_r$.

    \label{prop:suboptimal-truth}
\end{proposition}

The case for utility-aware approaches in games is further strengthened by the general volatility of the strategy-utility landscape. In particular,  we see that for a wide characterization of games, the probability that a uniformly sampled strategy is near another strategy with a \emph{different} best response increases with the dimensionality of the game. This phenomenon points to why, in some instances, small perturbations in state predictions result in completely different utilities. 
\begin{definition}
    \label{def:game_seq}
    Let $A, B$ be two players in a bi-matrix game $G= (\textbf{A, B})$. For player $A$, we define the best response region for each action $b$ of player $B$ to be $R_b^A := \{\alpha \in \Delta({\mathcal{A}})| b \in BR(\alpha)\}$.
\end{definition}

\begin{definition}
    \label{def:delta_reg}
    We say that a particular strategy $\alpha \in R^A_b$ is \textbf{$\delta-$regional} if $\exists \alpha' \in R^A_{b'}$ with $b \neq b'$ such that $\|\alpha - \alpha'\|_2 \leq \delta$.
\end{definition}
In the above definitions, we formally describe best response regions (i.e., collections of strategies that elicit the same best response from a follower) and $\delta$-regional strategies (strategies that, despite belonging to a particular best response region, are very close to another best response region). We seek to understand the prevalence of $\delta-$regional strategies. 
\begin{assumption}
\label{assumption:game_sparse}
Let $\{G_N\}_{N = 2}^{\infty}$ be a sequence of two-player bi-matrix games
in which both players have $N$ actions,
indexed by $[N] = \{1, \dots, N\}$.
For each action $b \in [N]$ for player $B$, the best-response region for player $A$, $R_{b}^{A,(N)}$, is such that (1) $R_{b}^{A,(N)} \neq \emptyset$ and (2) for all $\alpha \in R_{b}^{A,(N)}$, we have $\alpha_b \ge 2\delta_N$, where $\alpha_b$ denotes the probability that player $A$'s mixed strategy $\alpha$ places on its $b$-th action and $\{\delta_N\}$ is a sequence of positive numbers.
\end{assumption}

We briefly remark on the strength of these conditions. Condition~(1) simply requires that no player $B$ action is strictly dominated---that is, every action in the player $B$'s action set is a best response to some player $A$ mixed strategy. This is a mild non-degeneracy requirement.

Condition~(2) is more substantive but natural in a broad class of games. It asks that whenever player $B$'s best response is action $b$, player $A$ places at least probability $2\delta_N$ on their own action $b$. Since the numbering of actions is arbitrary, the condition need only hold under \emph{some} permutation. In security game formulations, for instance, the permutation pairs each attacker action $b$ with a defender action $\sigma(b)$, and the condition reduces to requiring that each attacked cell is driven by the defender's commitment to patrolling some corresponding location---a mild structural requirement.
% Condition~(2) is more substantive but remains natural in a broad class of games. Since both players' action spaces are identified with $[N]$, it asks that whenever the player $B$'s best response is action $b$, player $A$ is placing at least some probability on its own action $b$. Crucially, the condition imposes only a minimal coupling between player $A$'s support and player $B$'s induced best response. Note that the numbering of actions is arbitrary, and as such, for this condition to hold in any game, there simply needs to exist a permutation on the numbering of actions where the condition is satisfied. In security games, for instance, an attacker targets cell $b$ precisely because the defender is concentrated elsewhere; the permutation relabels actions so that attacker action $b$ is paired with defender action $\sigma(b)$, and the condition reduces to requiring that when the attacker strikes cell $b$, the defender is allocating at least some probability to the corresponding cell $\sigma(b)$---that is, each attack is driven by the defender's commitment to patrolling a particular other location.

\begin{theorem}\label{Thm:volatility}
Suppose the sequence of games $\{G_N\}_{N=2}^{\infty}$ satisfies
Assumption~\ref{assumption:game_sparse} with $\delta_N = c/N$ for some $c \in (0, \tfrac{1}{2}]$. For each $N$, let $\alpha \sim \mathrm{Unif}(\Delta(\mathcal{A}))$
in game $G_N$, then:
\[
  \mathbb{P}\!\big(\alpha \text{ is } \delta_N\text{-regional}\big)
  \;\geq\; 1 - \left(1 - \frac{c\sqrt{2}}{N}\right)^{\!N-1}.
\]
\end{theorem}

Theorem~\ref{Thm:volatility} illustrates how volatility naturally arises in games. For games satisfying Assumption~\ref{assumption:game_sparse}, a uniformly sampled leader strategy is $\delta_N$-regional with probability at least $1 - \left(1 - \frac{c\sqrt{2}}{N}\right)^{N-1}$. This bound increases with $N$ and has the non-vanishing limit $1 - e^{-c\sqrt{2}}$, which is approximately $0.507$ at $c = \tfrac{1}{2}$. As the dimensionality of the game grows, small perturbations to the leader's strategy may therefore suffice to change the follower's response. This susceptibility to downstream disruption is precisely what motivates a decision-focused approach.
\begin{remark}
Theorem~\ref{Thm:volatility} bounds the prevalence of near-boundary strategies in strategy space; it makes no statement about utilities. Two steps connect it to utilities: a small change in the prediction can move the induced strategy across a best-response boundary, and where adjacent regions carry different leader utilities the realized utility then can change sharply. Games with such utility differences across adjacent regions are possible. 
\end{remark}

\section{End to End Predictions in Strategic Settings}
\label{sec: e2e}
% Given the non-monotonicity and potential volatility of the payoff landscape with respect to the accuracy of predictions, we move to consider a new paradigm for developing predictions within strategic environments. This paradigm is what we call ``\textbf{End-To-End}'' (ETE) prediction selection. As a warm-up, we illustrate how to explicitly incorporate downstream utility in the finite state action setting.
The results of Section~\ref{sec: limits} establish that no accuracy-based 
training objective can serve as a reliable proxy for strategic utility. This 
motivates a fundamentally different approach to learning in the SDFL setting, 
which we call \textbf{End-To-End} (ETE) prediction selection: rather than 
optimizing for predictive accuracy in isolation, ETE directly targets the 
downstream strategic objective. We present this framework as an initial 
algorithmic response to SDFL --- a proof of concept demonstrating tractability 
rather than a definitive solution, in the same spirit as early end-to-end 
approaches in single-agent DFL that preceded formal convergence 
analysis~\cite{donti2017task, wilder2019melding}.

Formally, let $f_\theta: \mathcal{X} 
\to \mathcal{S}$ be a parameterized model mapping observable features $x \in 
\mathcal{X}$ (derived from dataset $\mathcal{D}$) to a predicted state $f_\theta(x) 
\in \mathcal{S}$. The ETE problem is to find parameters $\theta^*$ that maximize 
the leader's realized utility at equilibrium as per \eqref{eq:sdfl_objective}. This objective is challenging to optimize directly: the true state, $s_r$, is unobserved at 
training time, and the mapping $\theta \mapsto U_i(\cdot)$ passes through an 
equilibrium computation that is in general non-differentiable in $\theta$.
% \[
% \theta^* = \arg\max_\theta\; \mathbb{E}_{x \sim \mathcal{D}}\!\left[ 
% U_i\!\left(\alpha^*\!\left(f_\theta(x)\right),\; s_r\right)\right],
% \]
% where $\alpha^*(s) = (\alpha_i^*(s),\, \alpha_{-i}^*(s))$ denotes the 
% equilibrium joint strategy profile when agents play assuming state $s$, 
% and $s_r$ is the fixed true state against which utility is realized.

In the Stackelberg setting, this non-differentiability is concrete: 
$\bar{\alpha}_i^{f_\theta(x)}$ is the solution to a linear program (LP) whose 
payoff matrices are determined by $f_\theta(x)$, and the LP argmax is 
non-smooth. Moreover, by Proposition~\ref{prop:non-monotone} and 
Theorem~\ref{Thm:volatility}, the objective is neither monotone in prediction 
accuracy nor smooth in $\theta$, so optimizing a supervised proxy is 
unreliable. These challenges motivate a two-stage development. We first 
consider the case where $\mathcal{S}$ is discrete and finite: here, utility 
can be evaluated exactly for any candidate prediction, which isolates the 
algorithmic question of how to search over predictions by strategic utility. 
We then turn to the general continuous setting, where both sources of 
difficulty are present.

\subsection{Warm-up: End to End State Prediction with Discrete, Finite State Spaces}

A core principle of ETE is that candidate predictions should be evaluated
by the strategic utility they induce, not solely by statistical accuracy.
A good ETE algorithm must therefore explore candidate predictions and filter
them by downstream performance. We begin by showing this explore-then-filter
principle is sound. In the discrete, finite-state setting, bandit-style
approaches such as successive elimination can enumerate candidate predictions
and evaluate each by its induced utility. We show that the cumulative cost
of exploring suboptimal predictions is sublinear in $T$, even when the
equilibrium induced by each candidate must itself be learned (provided
per-state equilibrium computation is also sublinear in $T$).

\begin{theorem}
     Consider a game $G$ with a finite set of exogenous states, and let $s_r$ denote the true realized state. Let $\alpha_{i,t}^{s}$ denote the strategy selected by player $i$ at time step $t$ given state prediction $s$ produced by algorithm $\mathcal{ALG}$, and let $(\bar{\alpha}_{i}^s, \bar{\alpha}_{-i}^s)$ be the unique equilibrium joint strategy profile given prediction state $s$. Suppose utilities are bounded, i.e.\ $U_i(\cdot, \cdot, s) \in [0, 1]$ for all $s$. If
    \[
        \underset{\mathcal{ALG}(s)}{\mathbb{E}}\!\left[\sum_{t=1}^T \bigl|U_i(\alpha_{i,t}^{s}, \alpha_{-i,t}^{s}, s') - U_i(\bar{\alpha}_{i}^{s}, \bar{\alpha}_{-i}^{s}, s')\bigr|\right] \in o(T) \quad \forall\, \text{fixed } (s, s'),
    \]
    then $\mathcal{ALG'}$, which combines $\mathcal{ALG}$ with successive elimination to determine an evaluation state, yields:
    \[
    \underset{\mathcal{ALG'}}{\mathbb{E}}\!\left[\sum_{t=1}^T \left|U_i(\alpha_{i,t}^{s_t}, 
    \alpha_{-i,t}^{s_t}, s_r) - \max_{s \in \mathcal{S}}\, 
    U_i(\bar{\alpha}_{i}^{s}, \bar{\alpha}_{-i}^{s}, s_r)\right|\right] \in o(T).
    \]
    \label{thm:no-regret}
\end{theorem}

\subsection{End-to-End Prediction in General Settings}

The warm-up setting assumes a finite, discrete state space and allows repeated evaluations of the same state. These luxuries rarely hold in practice. In general applications, the state space is continuous or combinatorially large, and we are given a dataset from which we are to return a single model for state prediction. In this paradigm, we are therefore unable to directly apply bandit-style enumeration. We thus seek an approach that retains the core principle of evaluating candidate predictions via their downstream strategic utility within a single-pass learning framework.

The theory above presents two concrete challenges for optimizing in the SDFL domain. Firstly, Propositions~\ref{prop:non-monotone} and~\ref{prop:suboptimal-truth} establish that the mapping $\theta \mapsto U_i^{\text{realized}}$ is fundamentally misaligned with accuracy-based surrogates. This renders approaches such as gradient descent on prediction error alone insufficient. Secondly, Theorem~\ref{Thm:volatility} establishes that strategies close to a best-response boundary grow increasingly prevalent as game dimensionality increases; where adjacent regions carry different leader utilities, the utility landscape is discontinuous across the boundary. We therefore derive a heuristic tailored to the discontinuous optimization landscape induced by strategic responses in SDFL problems.

 Algorithm~\ref{Alg:ETE} addresses these challenges by combining three ideas. First, we perform \emph{intra-region optimization} within regions where the strategic response is fixed, using KKT-directed gradients that exploit differentiability within a fixed best-response zone to directly optimize through the follower's best response. Within these strategic epochs, the Huber regularisation weight $\lambda_e$ is annealed from $\lambda_{\max}$ to $\lambda_{\min}$, keeping prediction quality grounded in early KKT steps when the strategic gradient would otherwise dominate and destabilize training. Second, we perform \emph{inter-region exploration} to optimize across best-response regions via supervised updates, which shift the induced leader strategy across zone boundaries. A strategic checkpoint filter, motivated by the explore-then-filter principle of Theorem~\ref{thm:no-regret}, selects across candidates produced by both mechanisms to navigate the globally non-convex and discontinuous landscape.

\paragraph{KKT-directed gradients.}
Given a predicted state $s_{pred}$, the leader's optimal strategy $\bar{\alpha}_i$ solves a linear program whose KKT conditions form a differentiable system in $\bar{\alpha}_i$ and the payoff parameters. For a fixed follower pure action $a_{-i}$, implicit differentiation yields $\partial \bar{\alpha}_i / \partial s_{pred}$, enabling backpropagation of strategic utility to $\theta$. However, as Theorem~\ref{Thm:volatility} predicts, near best-response boundaries the follower's action switches discontinuously, and the KKT gradient—derived with $a_{-i}$ fixed—becomes unreliable. This motivates alternation: odd regimes apply supervised updates with Gaussian input perturbations to traverse $\Delta(\mathcal{A}_i)$ and escape boundary regions, while even regimes apply $\mathcal{U}_{\text{KKT}}$ where the gradient is well-defined.

\paragraph{Strategic checkpoint filter.}
At the close of every regime of $R$ epochs---whether a supervised or KKT regime---the current parameters $\theta$ are evaluated on a held-out set $\mathcal{D}_{\text{eval}}$ via the strategic utility function $\mathcal{U}$, and accepted only if they improve upon the best recorded utility $u^*$; otherwise training reverts to $\theta^*$. This mechanism is motivated in part by Theorem~\ref{thm:no-regret}: we identify model parameters by evaluating them for their downstream strategic utility.

\begin{remark}\label{rem:checkpoint}
The checkpoint filter selects a model whose expected strategic utility is within $\varepsilon$ of the best candidate: $\mathcal{U}(\theta^*) \geq \max_{1 \leq k \leq K} \mathcal{U}(\theta_k) - \varepsilon$ with probability at least $1 - \delta$, provided $|\mathcal{D}_{\text{eval}}| \geq 2\ln(2K/\delta)/\varepsilon^2$. This follows from standard arguments using Hoeffding's inequality and a union bound over the $K = \lceil T/R \rceil$ checkpoints.
\end{remark}

\begin{algorithm}[tb]
   \caption{Hybrid ETE: Alternating Predictive/Strategic Regimes}
   \label{Alg:ETE}
\begin{algorithmic}[1]
\small
   \STATE {\bfseries Input:} $\mathcal{D}$, $f_\theta$, $T$, $R$, $\rho$, $\sigma$, $\lambda_{\max}$, $\lambda_{\min}$, $E_{\text{anneal}}$
   \STATE $\mathcal{D}_{\text{grad}}, \mathcal{D}_{\text{eval}} \leftarrow \text{Split}(\mathcal{D}, \rho)$; \; $\theta^* \leftarrow \theta$; \; $u^* \leftarrow -\infty$; \; $e_{\text{kkt}} \leftarrow 0$
   \FOR{$r = 1$ {\bfseries to} $\lceil T/R \rceil$}
      \STATE $\theta \leftarrow \theta^*$
      \FOR{$e = 1$ {\bfseries to} $R$}
         \STATE \textbf{if} $r$ odd: $\theta \leftarrow \theta - \eta \nabla_\theta \mathcal{L}_{\text{pred}}(f_\theta(X_{\mathcal{B}} + \epsilon), Y_{\mathcal{B}})$, \; $\epsilon \sim \mathcal{N}(0,\sigma^2 I)$
         \STATE \textbf{else}: $\lambda_e \leftarrow \lambda_{\max} - (\lambda_{\max} {-} \lambda_{\min}) \cdot \min\!\left(\tfrac{e_{\text{kkt}}}{E_{\text{anneal}}}, 1\right)$; \; $e_{\text{kkt}} \mathrel{+}= 1$; \\
         \hspace{2.45em} $\theta \leftarrow \theta - \eta \nabla_\theta \bigl[-\mathcal{U}_{\text{KKT}}(f_\theta, \mathcal{B}) + \lambda_e \, \mathcal{L}_{\text{pred}}(f_\theta(X_{\mathcal{B}}), Y_{\mathcal{B}})\bigr]$
      \ENDFOR
      \STATE \textbf{if} $\mathcal{U}(f_\theta, \mathcal{D}_{\text{eval}}) > u^*$: \; $u^* \leftarrow \mathcal{U}(f_\theta, \mathcal{D}_{\text{eval}})$, \; $\theta^* \leftarrow \theta$
   \ENDFOR
   \STATE \textbf{return} $\theta^*$
\end{algorithmic}
\end{algorithm}

\subsection{Experimental Evaluation: Security Games}
Building on game-theoretic approaches to security
\cite{pita2008armor,fang2015securitygreen,krever2025guard}, we instantiate two
Stackelberg security games on an $N\!\times\!N$ grid $G$: a conservation game
using GPS tracks of 15 African elephants in Etosha National Park
\cite{getz2018etosha}, and an infrastructure game using ridership data from 123
Manhattan subway stations (Jan 2022--Dec 2024) \cite{mta2025ridership}. In both, a
defender commits to a mixed strategy $\alpha_i\in\Delta(G)$ over patrol
allocations, and an attacker who observes this commitment chooses
$\alpha_{-i}\in\Delta(G)$ \cite{yang2014adaptive}. The exogenous state
$s\in\mathbb{Z}_+^{N}$ captures the value at risk per cell: uncontested attacks
cost the defender utility proportional to the cell value $n$ (and reward the
attacker proportionally), successful interdiction rewards the defender
proportionally and penalizes the attacker by a fixed amount, and unvisited cells
contribute zero. Full details are in Appendix~\ref{sec:app-exp-details}.
% Building upon a large body of work that has studied the implications of game theoretic approaches within the security domain \cite{pita2008armor} \cite{fang2015securitygreen} \cite{krever2025guard}, we instantiate two versions of security games: a conservation security game (drawn from a study tracking 15 African elephants within Etosha National Park published on the Movebank repository \cite{getz2018etosha}) and an infrastructure security game (drawn from New York Metropolitan Transit Authority records \cite{mta2025ridership}, covering nine Manhattan stations from 1 January 2022 to 31 December 2024). In both instantiations, the primary agent is a defender distributing patrol resources over an $N \times N$ grid $G$, with strategy $\alpha_i \in \Delta(G)$; the attacker's strategy $\alpha_{-i} \in \Delta(G)$ similarly distributes attack resources over $G$. The exogenous state $s \in \mathbb{Z}_+^{N}$ is a vector whose entries quantify the value at risk in each cell. In each cell, if the attacker is uncontested the defender incurs negative utility proportional to the cell value $n$ while the attacker gains proportional utility; if the defender interdicts, the defender gains utility proportional to $n$ and the attacker suffers a fixed penalty; otherwise the cell contributes zero utility to either player. Consistent with previous work in this domain \cite{yang2014adaptive}, attackers are assumed to observe the defender's committed strategy, so the defender optimizes for the Stackelberg strategy. Details on both settings are in the appendix. 

\subsubsection{Methodology}

% In both the conservation and infrastructure security games, the defender must compute a Stackelberg strategy based on the current state $s$, which is unknown. The defender, instead, has access to historical state observations and must \emph{predict} the state at evaluation time in order to construct the payoff matrices from which the equilibrium strategy is derived.

We frame the problem as supervised regression: given temporal and historical
features for each cell in $G$, predict the scalar value (elephant count or
ridership) at a future period. All methods share a common base neural network architecture; comprehensive details on implementation, feature
engineering, and training are in Appendix~\ref{sec:app-exp-details}.

\paragraph{Results.}
We compare four methods. \textbf{PTO} minimizes predictive loss and feeds the
resulting prediction directly to the Stackelberg LP. The three ETE variants
differ in how strategic information enters training. \textbf{ZOC}
(Zeroth-Order Checkpoint) trains on supervised loss but applies a strategic
checkpoint filter after each regime, retaining only parameters that improve
utility on $\mathcal{D}_{\text{eval}}$. \textbf{KKT} differentiates through the
Stackelberg LP via KKT conditions, backpropagating strategic utility directly to
model parameters. \textbf{Hybrid} (Algorithm~\ref{Alg:ETE}) alternates
supervised and KKT regimes with the checkpoint filter applied after each.
We evaluate on two axes: Mean Absolute Error
$\mathrm{MAE}=\frac{1}{N|\mathcal{T}|}\sum_{t\in\mathcal{T}}\sum_{i\in[N]}
|\hat{y}^i_t - y^i_t|$ and downstream strategic utility across all evaluation
instances. Table~\ref{tab:results_500} reports both metrics at 500 epochs---the
budget at which all methods converged---across both game environments.
Utility is averaged over all evaluation instances (396 daily instances
for conservation, 1098 for infrastructure) and 5 independent training
runs; MAE is computed analogously. Figure~\ref{fig:utility_over_time} shows the same utilities month by month over each test period.

\begin{figure}[!ht]
    \centering
    \subcaptionbox{Conservation Security Game (Etosha), test period February 2013 to March 2014.\label{fig:utility_over_time_gsg}}[0.48\textwidth]{%
        \includegraphics[width=\linewidth]{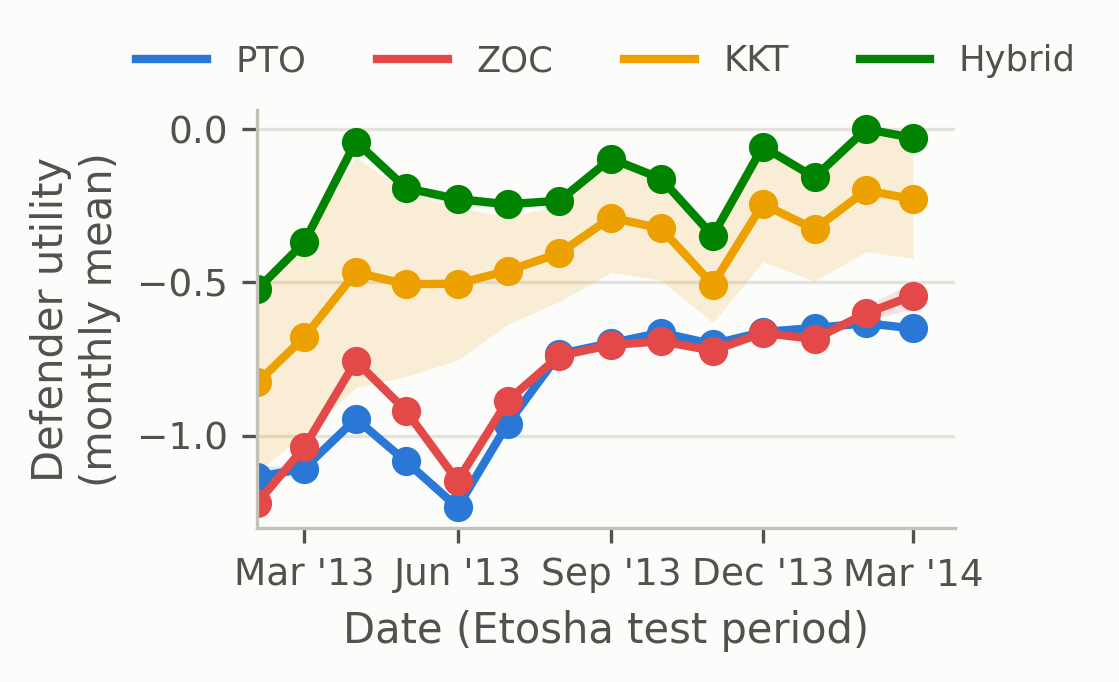}%
    }\hfill
    \subcaptionbox{Infrastructure Security Game (Manhattan subway), test period January to December 2024.\label{fig:utility_over_time_mta}}[0.48\textwidth]{%
        \includegraphics[width=\linewidth]{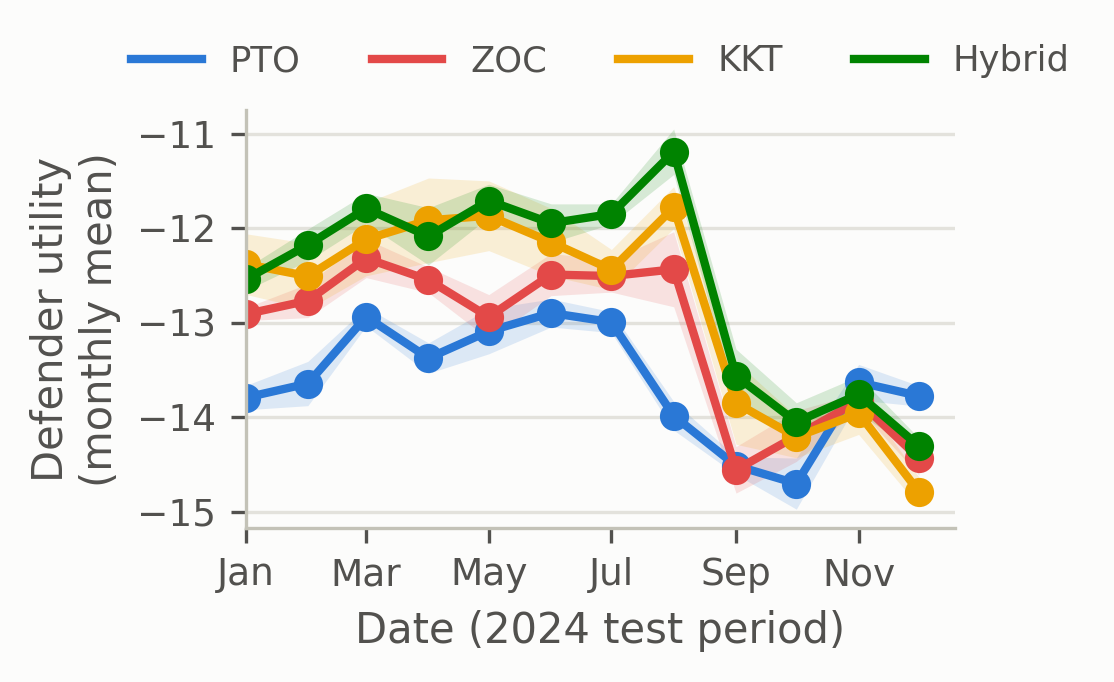}%
    }
    \caption{Defender utility over the test period for the four methods. Each point is the mean defender utility over the evaluation instances in one calendar month, averaged over the five training runs; the shaded band is one standard error across runs. Higher is better. Averaging each curve over all evaluation instances gives the utilities in Table~\ref{tab:results_500}.}
    \label{fig:utility_over_time}
\end{figure}

\begin{table}[ht]
\centering
\caption{ETE vs.\ PTO at 500 training epochs (mean $\pm$ std across 5 runs).}
\label{tab:results_500}
\begin{tabular}{llcccc}
\toprule
& & \multicolumn{2}{c}{\textbf{Conservation Security Game}} & \multicolumn{2}{c}{\textbf{Infrastructure Security Game}} \\
\cmidrule(lr){3-4} \cmidrule(lr){5-6}
& \textbf{Method} & \textbf{Avg.\ Utility} $(\uparrow)$ & \textbf{MAE} $(\downarrow)$ & \textbf{Avg.\ Utility} $(\uparrow)$ & \textbf{MAE} $(\downarrow)$ \\
\midrule
\multirow{3}{*}{\text{ETE}}
  & ZOC    & $-0.785 \pm 0.009$ & $0.118 \pm 0.006$ & $-13.160 \pm 0.355$ & $0.559 \pm 0.062$ \\
  & KKT    & $-0.404 \pm 0.483$ & $0.433 \pm 0.159$ & $-12.829 \pm 0.465$ & $1.104 \pm 0.467$ \\
  & Hybrid & $\mathbf{-0.173 \pm 0.013}$ & $0.458 \pm 0.107$ & $\mathbf{-12.578 \pm 0.383}$ & $0.736 \pm 0.028$ \\
\midrule
& PTO    & $-0.831 \pm 0.005$ & $\mathbf{0.110 \pm 0.005}$ & $-13.611 \pm 0.210$ & $\mathbf{0.409 \pm 0.009}$ \\
\bottomrule
\end{tabular}
\end{table}
% \begin{table}[ht]
% \centering
% \caption{ETE vs.\ PTO at 500 training epochs (mean $\pm$ std across 5 runs).}
% \label{tab:results_500}
% \begin{tabular}{ll@{\hspace{4pt}}lcc}
% \toprule
% \textbf{Game} & & \textbf{Method} & \textbf{Avg.\ Utility} $(\uparrow)$ & \textbf{MAE} $(\downarrow)$ \\
% \midrule
% \multirow{4}{*}{\shortstack[l]{Conservation\\Security Game}}
%   & \multirow{3}{*}{\text{ETE}} & ZOC    & $-0.785 \pm 0.008$ & $0.118 \pm 0.005$ \\
%   &                               & KKT    & $-0.404 \pm 0.432$ & $0.433 \pm 0.142$ \\
%   &                               & Hybrid & $\mathbf{-0.173 \pm 0.011}$ & $0.458 \pm 0.096$ \\
% \cmidrule(l){2-5}
%   &                               & PTO    & $-0.831 \pm 0.004$ & $\mathbf{0.111 \pm 0.005}$ \\
% \midrule
% \multirow{4}{*}{\shortstack[l]{Infrastructure\\Security Game}}
%   & \multirow{3}{*}{\text{ETE}} & ZOC    & $-11.984 \pm 0.183$ & $2.064 \pm 0.231$ \\
%   &                               & KKT    & $-11.514 \pm 0.149$ & $14.302 \pm 1.044$ \\
%   &                               & Hybrid & $\mathbf{-11.479 \pm 0.198}$ & $6.193 \pm 0.787$ \\
% \cmidrule(l){2-5}
%   &                               & PTO    & $-13.564 \pm 0.281$ & $\mathbf{1.368 \pm 0.021}$ \\
% \bottomrule
% \end{tabular}
% \end{table}

PTO, on average, achieves lower MAE than all ETE approaches, indicating superior predictive accuracy in the conventional supervised learning sense. However, this advantage does not translate into better strategic outcomes: all three ETE methods achieve higher defender utility than PTO in both domains. Critically, Hybrid achieves the best strategic utility across both game environments --- outperforming not only PTO but also ZOC and KKT in isolation --- demonstrating that combining supervised exploration, KKT-directed gradients, and the strategic checkpoint filter yields consistent gains that neither component delivers on its own.

\begin{figure}[H]
    \centering
    \subcaptionbox{Conservation Security Game: Elephant Distribution Predictions\label{fig:gsg_e2ewins}}{%
        \includegraphics[width=0.85\textwidth]{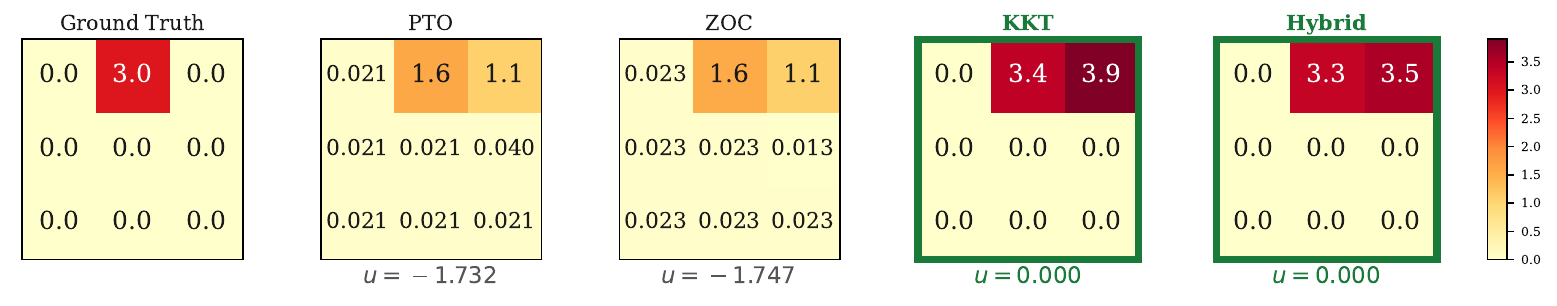}%
    }\\[0.1em]
    \subcaptionbox{Infrastructure Security Game: Patrol Strategies across top 9 stations\label{fig:mta_ptowins}}{%
        \includegraphics[width=0.85\textwidth]{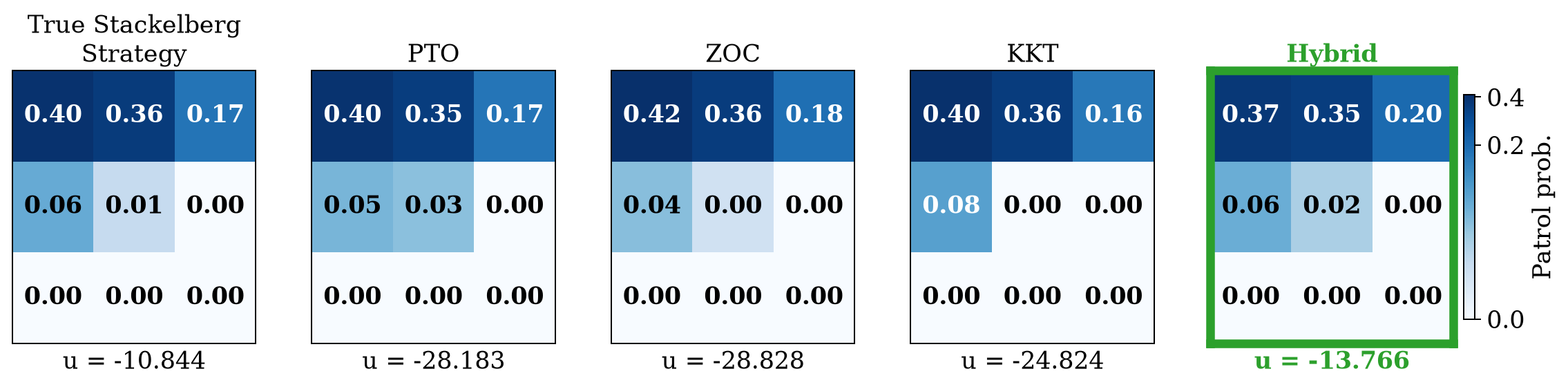}%
    }\\[0.1em]
    \subcaptionbox{Infrastructure Security Game: Ridership Predictions across 123 stations\label{fig:synthetic_hybrid}}{%  % fixed duplicate label + "wins wins"
        \includegraphics[width=0.85\textwidth]{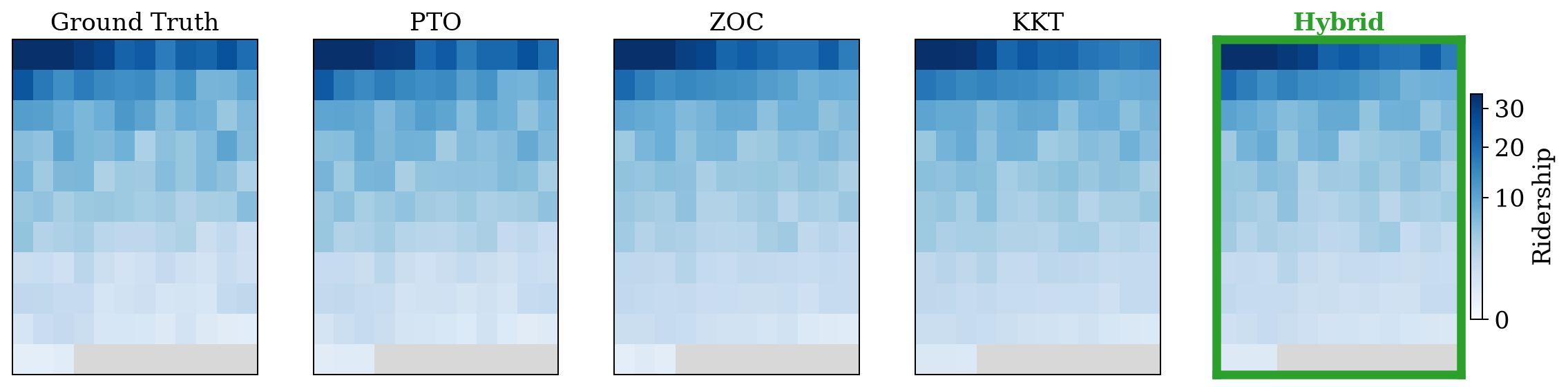}%
    }
    \caption{Instance-level state predictions for select evaluation periods.
    Winning method(s) shown in \textcolor{green!50!black}{\textbf{bold green}};
    utility $u$ reported below each prediction grid.}
    \label{fig:heatmap_comparison}
\end{figure}

Figure~\ref{fig:heatmap_comparison} illustrates the instance-level dynamics
underlying these aggregate results. The utility landscape is complex and
volatile---consistent with Theorem~\ref{Thm:volatility}---and no single method
universally dominates at the individual instance level. In instances where ETE
methods outperform PTO, the ETE predictions better capture strategically
relevant structure in the state, even at the cost of raw predictive accuracy. We illustrate such instances in the above figure. To probe the general landscape and gain an understanding of how in some cases we may see variation with PTO approaches, we include a  100 action synthetic setting on a grid with multiple modes in Appendix~\ref{sec:app-toy-exp}: here we illustrate how
different state representations produce distinct downstream strategic
consequences.

\FloatBarrier % all figures and tables of Section 4 are placed before the conclusion
\section{Conclusion}
\label{sec: conclusion}
In this work, we formalized \emph{strategic decision-focused learning} as a distinct problem
within the decision-focused learning paradigm, targeting settings where ML
predictions are consumed by strategically interacting agents. We showed that
prediction accuracy and strategic payoff can be non-monotonic, motivating a
departure from predict-then-optimize pipelines, and proposed an end-to-end
framework that incorporates downstream strategic utility directly into model
selection via KKT-directed gradients and strategic evaluation on a held-out
partition.
% In this work, we formalized \emph{strategic decision-focused learning} as a distinct problem within the broader decision-focused learning paradigm, addressing settings in which ML predictions are consumed by agents engaged in strategic interactions. We established theoretically that the relationship between prediction accuracy and strategic payoff can be non-monotonic, motivating a departure from the standard predict-then-optimize pipeline. To this end, we proposed an end-to-end framework that incorporates downstream strategic utility directly into model selection, coupling KKT-directed gradients with strategic evaluation on a held-out partition. 
% Experiments across two security game domains---wildlife conservation using elephant tracking data from Etosha National Park and infrastructure protection using MTA subway ridership data---confirmed that end-to-end training achieves superior defender utility despite higher supervised prediction error, corroborating our theoretical finding that accuracy and strategic utility can be misaligned.

Several important directions remain open. Our formulation assumes followers observe the true state; modeling settings in which followers must also learn the state from data introduces a second layer of strategic complexity worth exploring. More broadly, precisely characterizing the conditions under which equilibrium-aware end-to-end approaches outperform predict-then-optimize, and developing scalable algorithms for this regime, are promising avenues for future work.

\clearpage % flush the pending floats (Figures 2 and 3) before the references
\bibliographystyle{icml2026}  
\bibliography{NeurIPS}

%%%%%%%%%%%%%%%%%%%%%%%%%%%%%%%%%%%%%%%%%%%%%%%%%%%%%%%%%%%%
\newpage
\appendix
\section*{Appendix}
\label{appendix}
\section*{Limitations and Broader Impact}

\paragraph{Limitations.}
Our formulation assumes that followers observe the true state of the
world; in practice, followers may also need to learn the state from data,
introducing a second layer of strategic complexity that our framework does
not capture. We provide a local non-monotonicity result motivating
end-to-end training, but a precise characterization of when
equilibrium-aware approaches outperform predict-then-optimize remains
open. 

\paragraph{Broader impact.}
Strategic decision-focused learning is motivated by settings where a
decision-maker deploys ML predictions against agents who respond
strategically — a structure common in security, infrastructure allocation,
and market design. Improved patrol allocation in conservation and transit
safety are positive applications, but the same framework could in
principle be used in ways that
raise fairness and civil-liberties concerns.  

\section{Proofs of Theoretical Results}

\begin{proposition*}{~\ref{prop:non-monotone}}
 There exists a game instance with true state $s_r$, and predictions $s_{p,1}$ and $s_{p,2}$, wherein $s_{p,1}$ is strictly more accurate than $s_{p,2}$ under every accuracy metric $\phi$, yet player $i$'s utility at the equilibrium induced by $s_{p,1}$ is strictly lower than that induced by $s_{p,2}$:
\[
U_i\!\left(\alpha^*(s_{p,1}),\, s_r\right)
\;<\;
U_i\!\left(\alpha^*(s_{p,2}),\, s_r\right),
\]
where $\alpha^*(s) = (\alpha_i^*(s),\, BR_{-i}(\alpha_i^*(s), s_r))$ denotes the equilibrium joint strategy profile when the primary plays assuming state $s$, the rest of the agents see the true state $s_r$, and utility is evaluated at $s_r$.

\end{proposition*}
\begin{proof}[Proof of Proposition ~\ref{prop:non-monotone}]
We consider a Stackelberg game between two players, $A$ and $B$, in which player $A$ leads. Let $\mathcal{S} :=$
\[\left\{
s_0 = \begin{bmatrix} 1 \\ 0 \\ 1 \\ 1 \\ 3 \\ 1 \\ 1 \\ 0 \\ 1 \end{bmatrix}, \quad
s_1 =\begin{bmatrix} 1 \\ 1 \\ 1 \\ 1 \\ 1 \\ 1 \\ 1 \\ 1 \\ 1 \end{bmatrix}, \quad
s_2 = \begin{bmatrix} 0 \\ 0 \\ 0 \\ 0 \\ 9 \\ 0 \\ 0 \\ 0 \\ 0 \end{bmatrix}
\right\}\]

Let $\mathcal{A} = \mathcal{B} =  \{x \in \mathbb{R}^9: \sum\limits_{i = 1}^9 x_i =1, x_i \geq 0 \; \forall i \}$. We define $\alpha \in \mathcal{A}$ and $\beta \in \mathcal{B}$ to be strategies by players $A$ and $B$ respectively. We denote the utility functions by the following: 
\[
 U_A(\alpha, \beta, s) = \sum\limits_{i = 1}^9 s_i \beta_i (2\alpha_i - 1);\quad U_B(\alpha, \beta, s) = \sum\limits_{i = 1}^9 s_i \beta_i (1- 2\alpha_i)
\]

Given a state $s \in \mathcal{S}$, player $B$'s payoff is linear in $\beta$, so their best response is always a vertex of $\mathcal{B}$, i.e., a pure strategy $e_k$ for some $k \in \{1, \ldots, 9\}$.

Consider the following set of linear programs defined for each possible best response action of player $B$:
\begin{align*}
\alpha^{(s,k)} = \argmax\limits_{\alpha \in \mathcal{A}} \quad & s_k (2\alpha_k - 1) \\
\text{s.t.} \quad & s_k(1 - 2\alpha_k) \geq s_j(1 - 2\alpha_j) \quad \forall \; j \neq k
\end{align*}
The constraint ensures that action $k$ is indeed player $B$'s best response under $\alpha$. Note that player $A$'s Stackelberg strategy given state $s$ is obtained by solving the linear programs and selecting the linear program with the highest value:
\[
\alpha^s = \argmax\limits_{k} \left\{ U_A(\alpha^{(s, k)}, e_k, s) : k \in \{1, \ldots, 9\} \right\}
\]
with Stackelberg value:
\[
 V^*(s) = \max\limits_{k} \left\{ U_A(\alpha^{(s, k)}, e_k, s) : k \in \{1, \ldots, 9\} \right\}
\]
Solving the above optimization problems for Stackelberg strategies given predictions of $s_1$ and $s_2$ yields:
$\alpha^{s_1} =\begin{bmatrix} \frac{1}{9} \\ \frac{1}{9} \\ \frac{1}{9} \\ \frac{1}{9} \\ \frac{1}{9} \\ \frac{1}{9} \\ \frac{1}{9} \\ \frac{1}{9} \\ \frac{1}{9} \end{bmatrix}$ and for the state $s_2$ we have a set of possible strategies, i.e., $\alpha^{s_2} \in A^{s_2} =\{\alpha \in \mathcal{A} : \alpha_5 > \frac{1}{2}\}$.

Let $s_0$ be the real state and $s_1$ and $s_2$ be predictions. We calculate the utility of the Stackelberg strategies under the state $s_0$ and find that:

\paragraph{For Strategy $\alpha^{s_1}$:} Recall $\alpha^{s_1} = \frac{1}{9}\mathbf{1}$. Under $s_0$, player $B$'s payoff 
from playing action $i$ is $s_{0,i}(1-2\alpha^{s_1}_i) = \frac{7}{9}s_{0,i}$, which is maximized at $i=5$ 
(where $s_{0,5}=3$). Hence the best response is $e_5$, giving player $A$ the utility:
\[
U_A(\alpha^{s_1}, e_5, s_0) = 3\!\left(\frac{2}{9} - 1\right) = -\frac{7}{3}.
\]

\paragraph{Worst-case strategy in $A^{s_2}$.} Recall $A^{s_2} = \{\alpha \in \Delta(\mathcal{A}) : \alpha_5 > \frac{1}{2}\}$. For any $\alpha \in A^{s_2}$, we assess player $B$'s best response under $s_0$. Actions $i \in \{2, 8\}$ satisfy $s_{0,i} = 0$, so $U_B(e_i, \alpha, s_0) = 0$ regardless of $\alpha$. Action $i = 5$ yields $U_B(e_5, \alpha, s_0) = 3(1 - 2\alpha_5) < 0$ since $\alpha_5 > \frac{1}{2}$. For the remaining actions $i \in \{1, 3, 4, 6, 7, 9\}$, note that $\sum_{i \neq 5} \alpha_i < \frac{1}{2}$, so at least one such action corresponds to $\alpha_i < \frac{1}{2}$, giving $U_B(\alpha, e_i, s_0) = 1 - 2\alpha_i > 0$. Since these actions offer strictly positive utility to player $B$, player $B$'s best response lies in $\{e_k : k \in \{1,3,4,6,7,9\}\}$ for any $\alpha \in A^{s_2}$ with Player $A$'s resulting utility as $U_A = 2\alpha_k - 1$. With this we note that
\[
\inf_{\alpha \in A^{s_2}} U_A(\alpha, \beta^*(\alpha), s_0) = -1.
\]
\paragraph{Comparison.} Even in the worst case over $A^{s_2}$, player $A$'s utility is strictly greater 
than under $\alpha^{s_1}$:
\[
\inf_{\alpha \in A^{s_2}} U_A(\alpha, \beta^*(\alpha), s_0) = -1 \;>\; -\frac{7}{3} = U_A(\alpha^{s_1}, e_5, s_0).
\]

The sorted error vectors, i.e.\ the absolute errors $|s_1 - s_0|$ and $|s_2 - s_0|$ arranged in decreasing order, are $(2,1,1,0,0,0,0,0,0)$ and $(6,1,1,1,1,1,1,0,0)$, so every top-$k$ partial sum of the former is at most half the latter's, and thus by \cite[Ch. IV]{bhatia1997matrix}:
\[
\|s_1 - s_0\| \;\leq\; \tfrac{1}{2}\,\|s_2 - s_0\| \;<\; \|s_2 - s_0\|
\]
for every symmetric norm simultaneously. Therefore, despite $s_2$ being a less accurate prediction of $s_0$ under every accuracy metric of Definition~\ref{def:accuracy_metric} (for instance $\|s_2 - s_0\|_2 = \sqrt{42}$ versus $\|s_1 - s_0\|_2 = \sqrt{6}$), the Stackelberg strategy consistent with the prediction $s_2$ outperforms the Stackelberg strategy induced by the more accurate prediction $s_1$. Strict preference extends to any separable loss $\sum_i g(|e_i|)$, where $e_i$ denotes the $i$-th coordinate of the error $s - s_0$, with $g$ strictly increasing (e.g., Huber, MSE).

%%%%%%%%%%%%%%%%%%%%%%%%%%%%%%%%%%%%%%%

\end{proof}

\begin{proposition*}{~\ref{prop:suboptimal-truth}}
There exists a game instance with true state $s_r$ and player $i$ such that there exists a prediction $s_p \in \mathcal{S}$ satisfying:
\[
U_i\!\left(\alpha^*(s_p),\, s_r\right) \;>\; U_i\!\left(\alpha^*(s_r),\, s_r\right),
\] where $\alpha^*(s) = (\alpha_i^*(s),\, BR_{-i}(\alpha_i^*(s), s_r))$ denotes the equilibrium joint strategy profile when the primary plays assuming state $s$, the rest of the agents see the true state $s_r$, and utility is evaluated at $s_r$.
\end{proposition*}
\begin{proof}[Proof of Proposition~\ref{prop:suboptimal-truth}]

We follow the proof structure and argument of Theorem 3.4 in \cite{handina2024understanding}. We work in the following setting: we have a two-player continuous game in which the primary agent (player~$i$) has strategy space from a compact set $\mathcal{A}_i \subseteq \mathbb{R}$ and the opposing agent (player~$-i$) has strategy space from a compact set $\mathcal{A}_{-i} \subseteq{\mathbb{R}}$, with utilities $U_i(\alpha_i, \alpha_{-i}, s)$ and $U_{-i}(\alpha_i, \alpha_{-i}, s)$ respectively. The exogenous state space $\mathcal{S} = \{0, 1\} $.  We assume the following.

\begin{assumption}
Assume the game defined on $U_i$ and $U_{-i}$ is strongly monotone on $\mathcal{A}_i \times \mathcal{A}_{-i}$. Further assume that
\begin{enumerate}
\item $U_i$ and $U_{-i}$ are jointly concave in $\alpha_i$ and $\alpha_{-i}$.
\item The gradient mappings $\nabla U_i$ and $\nabla U_{-i}$ exist and are well defined for all $(\alpha_i, \alpha_{-i})$. Furthermore, the gradient mappings are $L$-Lipschitz continuous in the joint action space.
\item The Nash equilibrium $(\alpha_{i}^{\star}, \alpha_{-i}^{\star}) \in \mathcal{A}_i \times \mathcal{A}_{-i}$ is in the interior with $\nabla_{\alpha_i} BR_{-i}(\alpha_{i}^{\star}, s_r) \neq 0$.
\item The unique Nash equilibrium $(\alpha_{i}^{\star}, \alpha_{-i}^{\star})$ is not Pareto-optimal: there exists $(\alpha_{i}^{\dagger}, \alpha_{-i}^{\dagger}) \in \mathcal{A}_i \times \mathcal{A}_{-i}$ such that $U_i(\alpha_{i}^{\dagger}, \alpha_{-i}^{\dagger}, s_r) > U_i(\alpha_{i}^{\star}, \alpha_{-i}^{\star}, s_r)$.
\end{enumerate}
\end{assumption}

Note the following: because we can achieve strictly higher utility for the jointly concave utility function $U_i$, we know that $\nabla U_i(\alpha_{i}^{\star}, \alpha_{-i}^{\star}, s_r) \neq 0$.

We define $\bar{U}_i(\alpha_i) := U_i(\alpha_i, BR_{-i}(\alpha_i, s_r), s_r)$ as the primary agent's utility when they play $\alpha_i$ and the opposing agent best-responds under the true state $s_r$. \cite{handina2024understanding} establishes that there exists a $v \in \mathbb{R}$ with $ v \cdot \nabla \bar{U}_i(\alpha_{i}^{\star}) < 0$ and a value $\delta > 0$ such that the following construction yields a Nash equilibrium with strictly higher utility for the primary agent. Specifically, setting $\alpha_i' = \alpha_{i}^{\star} - \delta v$, $\alpha_{-i}' = BR_{-i}(\alpha_i', s_r)$, and defining
\[
\tilde{\Theta} = \{ \alpha_i \in \mathcal{A}_i : ( \alpha_i - \alpha_i') \cdot v \leq 0 \}, \quad \Theta' = \{ \alpha_i \in \tilde{\Theta} : \nabla_{\alpha_i} U_i(\alpha_i', \alpha_{-i}', s_r) \cdot (\alpha_i' - \alpha_i) \geq 0 \},
\]
the restricted game on $\Theta' \times \mathcal{A}_{-i}$ admits $(\alpha_i', \alpha_{-i}')$ as a Nash equilibrium. By a Taylor expansion argument using the $L$-Lipschitz gradient assumption, one can choose $\delta$ small enough so that $\bar{U}_i(\alpha_i') > \bar{U}_i(\alpha_{i}^{\star})$, meaning the primary agent achieves strictly higher utility at this new equilibrium than at $(\alpha_{i}^{\star}, \alpha_{-i}^{\star})$. To complete the argument, for player $i$ let $s = 0$ correspond to $U_i(\alpha_i, \alpha_{-i}, 0) > -M \quad \forall (\alpha_i, \alpha_{-i})$ with $M \in \mathbb{R}_{+}$ and then for $s = 1$ let $U_i(\alpha_i, \alpha_{-i}, 1) = U_i(\alpha_i, \alpha_{-i}, 0) \quad \forall (\alpha_i, \alpha_{-i}) \in \Theta' \times \mathcal{A}_{-i}$ and $U_i(\alpha_i, \alpha_{-i}, 1) = -M$ otherwise. For player $-i$ we have $U_{-i}(\alpha_i, \alpha_{-i}, 1) = U_{-i}(\alpha_i, \alpha_{-i}, 0)$. In this simple setup, a prediction $(s_p)$ of state $1$ results in better utility for player $i$ even if the real state ($s_r)$ is $0$.

Such games exist; for instance, on $\mathcal{A}_i = \mathcal{A}_{-i} = [0,1]$ the game $U_i(\alpha_i, \alpha_{-i}, 0) = -\tfrac12(\alpha_i - \alpha_{-i})^2 + \alpha_{-i}$, $U_{-i}(\alpha_i, \alpha_{-i}, 0) = -\tfrac12\big(\alpha_{-i} - \tfrac{\alpha_i}{2} - \tfrac14\big)^2$ satisfies the assumption: it is strongly monotone with $\mu = \tfrac14$, both utilities are concave with affine gradients, $BR_{-i}(\alpha_i, 0) = \tfrac{\alpha_i}{2} + \tfrac14$ has derivative $\tfrac12 \neq 0$, the unique Nash equilibrium $(\tfrac12, \tfrac12)$ is interior, and $U_i(1, \frac{3}{4}, 0) = \frac{23}{32} > \tfrac12 = U_i(\tfrac12, \tfrac12, 0)$ and $U_{-i}(1, \frac{3}{4}, 0) = 0 = U_{-i}(\tfrac12, \tfrac12, 0)$. Consider the game restricted for $ \alpha_i$ to  $[\tfrac34, 1]$, we get $\alpha_i' = \tfrac34$, $\alpha_{-i}' = \tfrac58$ and $\bar{U}_i(\tfrac34) = \tfrac{79}{128} > \tfrac{1}{2} = \bar{U}_i(\tfrac12)$.

As a note, this proof, together with the establishment of $v$ from \cite{handina2024understanding}, can be extended to more dimensions by adding the assumption that $\nabla_{\alpha_i} BR_{-i}(\alpha_{i}^{\star}, s_r)$ has rank equal to the dimension of $\mathcal{A}_{-i}$. 

\end{proof}
\begin{theorem*}{~\ref{Thm:volatility}}
Suppose the sequence of games $\{G_N\}_{N=2}^{\infty}$ satisfies
Assumption~\ref{assumption:game_sparse} with $\delta_N = c/N$ for some $c \in (0, \tfrac{1}{2}]$. For each $N$, let $\alpha \sim \mathrm{Unif}(\Delta(\mathcal{A}))$
in game $G_N$, then:
\[
  \mathbb{P}\!\big(\alpha \text{ is } \delta_N\text{-regional}\big)
  \;\geq\; 1 - \left(1 - \frac{c\sqrt{2}}{N}\right)^{\!N-1}.
\]
\end{theorem*}

\begin{proof}[Proof of Theorem~\ref{Thm:volatility}]
We instantiate the geometric arguments of~\cite{MirzakhaniVondrak2017}
to our context.
The conditions of Assumption~\ref{assumption:game_sparse} ensure that
the best-response regions $\{R^{A,(N)}_b\}_{b \in [N]}$ form a
Knaster--Kuratowski--Mazurkiewicz (KKM) covering~\cite{KnasterKuratowskiMazurkiewicz1929}
of the simplex $\Delta(\mathcal{A})$.
Let $\mu_{N-1}$ denote the $(N{-}1)$-dimensional Lebesgue measure on $\Delta(\mathcal{A})$.

Define $S = \bigcup_{b \neq b'} (R^{A,(N)}_b \cap R^{A,(N)}_{b'})$
and its $\delta_N$-neighborhood
$S_{\delta_N} = \{ \alpha \in \Delta(\mathcal{A}) : \exists\, \alpha' \in S,\;
\|\alpha - \alpha'\|_2 \leq \delta_N \}$.
By definition, $\alpha$ is $\delta_N$-regional if and only if
$\alpha \in S_{\delta_N}$.
Set $\overline{R}^{A,(N)}_b = R^{A,(N)}_b \setminus S_{\delta_N}$,
so that
$\bigcup_{b \in [N]} \overline{R}^{A,(N)}_b = \Delta(\mathcal{A}) \setminus S_{\delta_N}$.
Let $\delta_N' = \delta_N\sqrt{2}$ and define the shifted regions
\[
\widetilde{R}^{A,(N)}_b
\;=\; \overline{R}^{A,(N)}_b - \delta_N' e_b
\;=\; \bigl\{\alpha - \delta_N' e_b : \alpha \in \overline{R}^{A,(N)}_b\bigr\}.
\]

The following claim adapts the shrink-and-shift
technique of~\cite[Lemma~3]{MirzakhaniVondrak2017} to our setting.

\begin{claim}\label{claim:shifted_disjoint}
The sets $\widetilde{R}^{A,(N)}_1, \ldots, \widetilde{R}^{A,(N)}_N$
are pairwise disjoint subsets of $(1 - \delta_N') \cdot \Delta(\mathcal{A})$.
\end{claim}

\begin{proof}[Proof of Claim~\ref{claim:shifted_disjoint}]
\emph{Disjointness.}\;
Suppose for contradiction that
$\widetilde{R}^{A,(N)}_b \cap \widetilde{R}^{A,(N)}_{b'} \neq \emptyset$
for some $b \neq b'$.
Then there exist $\alpha \in \overline{R}^{A,(N)}_b$ and
$\beta \in \overline{R}^{A,(N)}_{b'}$ with
$\alpha - \delta_N' e_b = \beta - \delta_N' e_{b'}$,
giving
$\|\alpha - \beta\| = \delta_N' \|e_b - e_{b'}\| = \delta_N'\sqrt{2} = 2\delta_N$.
The midpoint $m = \tfrac{1}{2}(\alpha + \beta)$ lies in $\Delta(\mathcal{A})$
by convexity, so $m \in R^{A,(N)}_\ell$ for some $\ell$
with $\ell \neq b$ or $\ell \neq b'$ (possibly both).
Assuming without loss of generality that $\ell \neq b$,
closedness of $R^{A,(N)}_b$ and $R^{A,(N)}_\ell$ yields a point
$\alpha' \in R^{A,(N)}_b \cap R^{A,(N)}_\ell \subseteq S$
on the segment from $\alpha$ to $m$, with
$\|\alpha' - \alpha\| \leq \|m - \alpha\| = \delta_N$.
This places $\alpha \in S_{\delta_N}$,
contradicting $\alpha \in \overline{R}^{A,(N)}_b$.

\emph{Containment.}\;
For any $\alpha \in \overline{R}^{A,(N)}_b$,
Assumption~\ref{assumption:game_sparse} gives
$\alpha_b \geq 2\delta_N > \delta_N' = \delta_N\sqrt{2}$,
so the shifted point $\alpha - \delta_N' e_b$ has all coordinates nonnegative
and summing to $1 - \delta_N'$.
Hence $\widetilde{R}^{A,(N)}_b \subseteq (1 - \delta_N') \cdot \Delta(\mathcal{A})$.
\end{proof}

Since each $\widetilde{R}^{A,(N)}_b$
is a translate of $\overline{R}^{A,(N)}_b$,
Claim~\ref{claim:shifted_disjoint} gives
\begin{align*}
\mu_{N-1}\bigl(\Delta(\mathcal{A}) \setminus S_{\delta_N}\bigr)
&= \sum_{b=1}^{N} \mu_{N-1}\bigl(\overline{R}^{A,(N)}_b\bigr)
= \sum_{b=1}^{N} \mu_{N-1}\bigl(\widetilde{R}^{A,(N)}_b\bigr) \\
&\leq \mu_{N-1}\bigl((1 - \delta_N') \cdot \Delta(\mathcal{A})\bigr)
= (1 - \delta_N\sqrt{2})^{N-1}\, \mu_{N-1}\bigl(\Delta(\mathcal{A})\bigr),
\end{align*}
and therefore
\[
\mathbb{P}\big(\alpha \text{ is } \delta_N\text{-regional}\big)
= \frac{\mu_{N-1}(S_{\delta_N})}{\mu_{N-1}(\Delta(\mathcal{A}))}
\geq 1 - (1 - \delta_N\sqrt{2})^{N-1},
\]
where $\alpha \sim \mathrm{Unif}(\Delta(\mathcal{A}))$.
Setting $\delta_N = c/N$ for $c \in (0, \tfrac{1}{2}]$, this becomes
\[
\mathbb{P}\big(\alpha \text{ is } \delta_N\text{-regional}\big)
\geq 1 - \left(1 - \frac{c\sqrt{2}}{N}\right)^{N-1}.
\]
\end{proof}

\begin{corollary}{\label{cor:volatility_limit}}
Under the hypotheses of Theorem~\ref{Thm:volatility},
\[
  \liminf_{N \to \infty}\;
  \mathbb{P}\!\big(\alpha \text{ is } \delta_N\text{-regional}\big)
  \;\geq\; 1 - e^{-c\sqrt{2}}.
\]
\end{corollary}
\begin{proof}
Since $\ln(1 - x) \leq -x$ for all $x \in (0,1)$,
\[
\left(1 - \frac{c\sqrt{2}}{N}\right)^{N-1}
\leq e^{-c\sqrt{2}(1 - 1/N)},
\]
and taking $N \to \infty$ yields the result.   
\end{proof}

\begin{theorem*}{~\ref{thm:no-regret}}
    Consider a game $G$ with a finite set of exogenous states, and let $s_r$ denote the true realized state. Let $\alpha_{i,t}^{s}$ denote the strategy selected by player $i$ at time step $t$ given state prediction $s$ produced by algorithm $\mathcal{ALG}$, and let $(\bar{\alpha}_{i}^s, \bar{\alpha}_{-i}^s)$ be the unique equilibrium joint strategy profile given prediction state $s$. Suppose utilities are bounded, i.e.\ $U_i(\cdot, \cdot, s) \in [0, 1]$ for all $s$. If
    \[
        \underset{\mathcal{ALG}(s)}{\mathbb{E}}\!\left[\sum_{t=1}^T \bigl|U_i(\alpha_{i,t}^{s}, \alpha_{-i,t}^{s}, s') - U_i(\bar{\alpha}_{i}^{s}, \bar{\alpha}_{-i}^{s}, s')\bigr|\right] \in o(T) \quad \forall\, \text{fixed } (s, s'),
    \]
    then $\mathcal{ALG'}$, which combines $\mathcal{ALG}$ with successive elimination to determine  an evaluation state, yields strategies satisfying
    \[
    \underset{\mathcal{ALG'}}{\mathbb{E}}\!\left[\sum_{t=1}^T \left|U_i(\alpha_{i,t}^{s_t}, 
    \alpha_{-i,t}^{s_t}, s_r) - \max_{s \in \mathcal{S}}\, 
    U_i(\bar{\alpha}_{i}^{s}, \bar{\alpha}_{-i}^{s}, s_r)\right|\right] \in o(T).
    \]
    % then $\mathcal{ALG'}$ which combines $\mathcal{ALG}$ with successive elimination to determine an evaluation state yields strategies satisfying
    % \[
    %     \underset{\mathcal{ALG'}}{\mathbb{E}}\!\left[\sum_{t=1}^T \bigl|U_i(\alpha_{i,t}^{s_t}, \alpha_{-i,t}^{s_t}, s_r) - U_i(\bar{\alpha}_{i}^{s_r}, \bar{\alpha}_{-i}^{s_r}, s_r)\bigr|\right] \in o(T).
    % \]
\end{theorem*}

\begin{proof}[Proof of Theorem ~\ref{thm:no-regret}]

\textbf{Step 1: Bandit Reduction.} Since $\mathcal{S}$ is finite, we treat each state $s \in \mathcal{S}$ as an arm in a multi-armed bandit problem. The goal of successive elimination is to identify a best evaluation state, i.e.\ a state $s$ under which $\mathcal{ALG}$ achieves the highest utility when evaluated at the true state $s_r$; by Proposition~\ref{prop:suboptimal-truth} this state need not be $s_r$ itself. The key difficulty is that within a single run of $\mathcal{ALG}(s)$, the rewards
\[
    X_t^s := U_i(\alpha_{i,t}^s, \alpha_{-i,t}^s, s_r)
\]
are correlated across time steps $t$, since $\mathcal{ALG}$ is adaptive. We resolve this by redefining what constitutes a single arm pull via an epoch structure.

\textbf{Notation.} For each state $s \in \mathcal{S}$ write
\[
    R^s := U_i(\bar{\alpha}_i^s, \bar{\alpha}_{-i}^s, s_r)
\]
for the arm's true value, the equilibrium utility of prediction $s$ evaluated at the true state; $R^s$ is free of any epoch index. Let $s^* \in \argmax_{s \in \mathcal{S}} R^s$ be a best evaluation state; it need not equal $s_r$. Write $\Delta_s := R^{s^*} - R^s \geq 0$ for the gaps at these values and $\Delta_{\min}$ for the smallest positive gap. (If no gap is positive, every state is optimal and only the exploration cost below is incurred.)

\textbf{Step 2: Epoch Structure.} Define epochs $k = 1, 2, \ldots$ with lengths $T_k = k$. 
Let $\mathcal{S}_k \subseteq \mathcal{S}$ denote the set of surviving arms entering epoch $k$, 
initialized to $\mathcal{S}_1 = \mathcal{S}$. In each epoch $k$, for every surviving arm 
$s \in \mathcal{S}_k$, we run a fresh independent instantiation of $\mathcal{ALG}(s)$ over 
$T_k$ rounds and compute the empirical reward
\[
    \hat{R}_k^s := \frac{1}{T_k}\sum_{t=1}^{T_k} U_i(\alpha_{i,t}^s, \alpha_{-i,t}^s, s_r).
\]
Since each epoch uses a fresh independent instantiation, $\hat{R}_k^s$ and $\hat{R}_{k'}^s$ 
are independent across epochs $k \neq k'$, recovering the independence across pulls that the 
bandit framework requires. Within-epoch correlations are absorbed into the single observation 
$\hat{R}_k^s$. Since $U_i(\cdot,\cdot,s) \in [0,1]$ by assumption, $\hat{R}_k^s \in [0,1]$, and any bounded random variable is sub-Gaussian.

The epoch rewards are not identically distributed. Applied over an epoch of length $T_k$, the hypothesis of the theorem (with $s' = s_r$) gives
\[
    \bigl|\mathbb{E}[\hat{R}_k^s] - R^s\bigr| \;\leq\; \frac{1}{T_k}\, \underset{\mathcal{ALG}(s)}{\mathbb{E}}\!\left[\sum_{t=1}^{T_k} \bigl|U_i(\alpha_{i,t}^{s}, \alpha_{-i,t}^{s}, s_r) - R^s\bigr|\right] \;=\; \frac{o(T_k)}{T_k},
\]
which vanishes as the epoch index $k$ grows. This is a property of $\mathcal{ALG}$ on the instance, independent of the horizon $T$; hence there is a constant $k_0$, entering the analysis only, with per-epoch bias at most $\Delta_{\min}/8$ for all $k \geq k_0$ and all $s \in \mathcal{S}$.

\textbf{Step 3: Successive Elimination over States.} We now apply the Successive Elimination 
algorithm of ~\cite{evendar2002pac} with arms $\mathcal{S}$. Recall that at each epoch 
$k$, arm $s$ produces the independent observation $\hat{R}_k^s \in [0,1]$. The 
empirical mean of arm $s$ after $t$ epochs is
\[
    \hat{p}_t^s := \frac{1}{t}\sum_{k=1}^{t} \hat{R}_k^s, \qquad \text{centered at } \mathbb{E}[\hat{p}_t^s] = \frac{1}{t}\sum_{k=1}^{t} \mathbb{E}[\hat{R}_k^s].
\]
The algorithm maintains a surviving set $\mathcal{S}_t \subseteq \mathcal{S}$ and at each epoch 
$t$ eliminates any arm $s$ satisfying
\[
    \hat{p}_t^{\max} - \hat{p}_t^s \;\geq\; 2\alpha_t, \qquad 
    \alpha_t = \sqrt{\frac{\ln(c|\mathcal{S}|t^2/\delta)}{t}},
\]
where $\hat{p}_t^{\max} = \max_{s \in \mathcal{S}_t} \hat{p}_t^s$, $c > 4$ is an absolute 
constant, and $\delta = 1/T$.

The probabilistic engine of Theorem~3 of ~\cite{evendar2002pac} is the tail bound $\Pr\bigl[|\hat{p}_t^s - \mathbb{E}[\hat{p}_t^s]| \geq \alpha_t\bigr] \leq e^{-\alpha_t^2 t} \leq \delta/(c|\mathcal{S}|t^2)$ with a union bound over epochs and arms. This is Hoeffding's inequality, whose requirements are independence and sub-Gaussian summands, both established in Step~2, and which centers the empirical mean at the average of the summands' means; identical distribution is not among its hypotheses. Hence with probability at least $1 - \delta$,
\[
    |\hat{p}_t^s - \mathbb{E}[\hat{p}_t^s]| \leq \alpha_t \qquad \text{for all } s \in \mathcal{S} \text{ and all } t \geq 1;
\]
call this event $E$.

What must be supplied is the mean control: that $\mathbb{E}[\hat{p}_t^s]$ is close enough to $R^s$. By Step~2,
\[
    \bigl|\mathbb{E}[\hat{p}_t^s] - R^s\bigr| \;\leq\; \frac{k_0}{t} + \frac{\Delta_{\min}}{8} \;\leq\; \frac{\Delta_{\min}}{4} \qquad \text{once } t \geq 8k_0/\Delta_{\min}.
\]
With this in hand, on $E$ the confidence intervals cover $R^s$ with radius $\alpha_t + \Delta_{\min}/4$: for $t \geq 8k_0/\Delta_{\min}$ the best arm is never eliminated by a suboptimal arm $s$, its empirical deficit being at most
\[
    \hat{p}_t^s - \hat{p}_t^{s^*} \;\leq\; -\Delta_s + \frac{\Delta_{\min}}{2} + 2\alpha_t \;<\; 2\alpha_t;
\]
and a suboptimal arm $s$ is eliminated once $4\alpha_t \leq \Delta_s - \Delta_{\min}/2$, for which $\alpha_t \leq \Delta_s/8$ suffices, i.e.\ within $K^* = O(\ln T/\Delta_{\min}^2)$ epochs. Before $t \geq 8k_0/\Delta_{\min}$ no elimination can occur for $T \geq e^{2k_0/\Delta_{\min}}$: an elimination requires $2\alpha_t \leq 1$, hence $t \geq 4\ln T \geq 8k_0/\Delta_{\min}$. An optimal arm may eliminate another optimal arm; this is harmless, and every surviving set contains an optimal arm.

The accounting holds at any finite $T$: the initial epochs before the bias settles cost $O(1)$ rounds; the elimination phase costs
\[
    N_{\mathrm{elim}} \;\leq\; |\mathcal{S}| \sum_{k \leq K^*} T_k \;\leq\; |\mathcal{S}|\, K^{*2} \;=\; o(T)
\]
rounds; and the failure event $E^C$ contributes at most $\delta T = 1$.

%%%
\textbf{Step 4: Bounding the Cumulative Suboptimality.} Split the horizon at the end of epoch $K^*$ and decompose
\begin{align*}
    &\underset{\mathcal{ALG'}}{\mathbb{E}}\!\left[\sum_{t=1}^T \left|U_i(\alpha_{i,t}^{s_t}, 
    \alpha_{-i,t}^{s_t}, s_r) - \max_{s \in \mathcal{S}}\,U_i(\bar{\alpha}_i^s, 
    \bar{\alpha}_{-i}^s, s_r)\right|\right] \\
    &\leq \underset{\mathcal{ALG'}}{\mathbb{E}}\!\left[\sum_{t=1}^T \left|\cdots\right| \,\Big|\, E\right] + \underset{\mathcal{ALG'}}{\mathbb{E}}\!\left[\sum_{t=1}^T 
    \left|\cdots\right| \,\Big|\, E^C\right] \cdot \mathbb{P}(E^C).
\end{align*}

\noindent\textit{Contribution from $E^C$.} Since $U_i \in [0,1]$, each summand is bounded by $1$, so this term is at most $T \cdot \mathbb{P}(E^C) \leq T\delta = 1$.

\noindent\textit{Rounds up to epoch $K^*$, on $E$.} These number at most $N_{\mathrm{elim}} = o(T)$, each contributing at most $1$.

\noindent\textit{Rounds after epoch $K^*$, on $E$.} By Step~3 every surviving arm $s$ satisfies $R^s = \max_{s' \in \mathcal{S}} R^{s'}$. If a single arm survives, the algorithm commits to it and, by the hypothesis of the theorem applied at that $s$,
\[
    \underset{\mathcal{ALG}(s)}{\mathbb{E}}\!\left[\sum_{t} \left|U_i(\alpha_{i,t}^{s}, 
    \alpha_{-i,t}^{s}, s_r) - R^s\right|\right] \in o(T),
\]
so the remaining rounds contribute $o(T)$. If several optimal arms survive, the epochs continue on them; each epoch on an optimal arm contributes $o(T_k)$ by the same hypothesis, uniformly over the finitely many arms, so these epochs contribute $o(T)$ in total.

\noindent Combining the three contributions, the total cumulative suboptimality of $\mathcal{ALG'}$ 
satisfies
\[
    \underset{\mathcal{ALG'}}{\mathbb{E}}\!\left[\sum_{t=1}^T \left|U_i(\alpha_{i,t}^{s_t}, 
    \alpha_{-i,t}^{s_t}, s_r) - \max_{s \in \mathcal{S}}\,U_i(\bar{\alpha}_i^s, 
    \bar{\alpha}_{-i}^s, s_r)\right|\right] \;\in\; o(T). 
\]

\end{proof}

\section{Experimental Details}\label{sec:app-exp-details}

\begin{figure*}[h]
    \centering
    \subcaptionbox{Etosha National Park with adaptive patrol-grid overlay \cite{getz2018etosha,esri2024imagery}.\label{fig:etosha_map}}{%
        \includegraphics[width=0.48\textwidth]{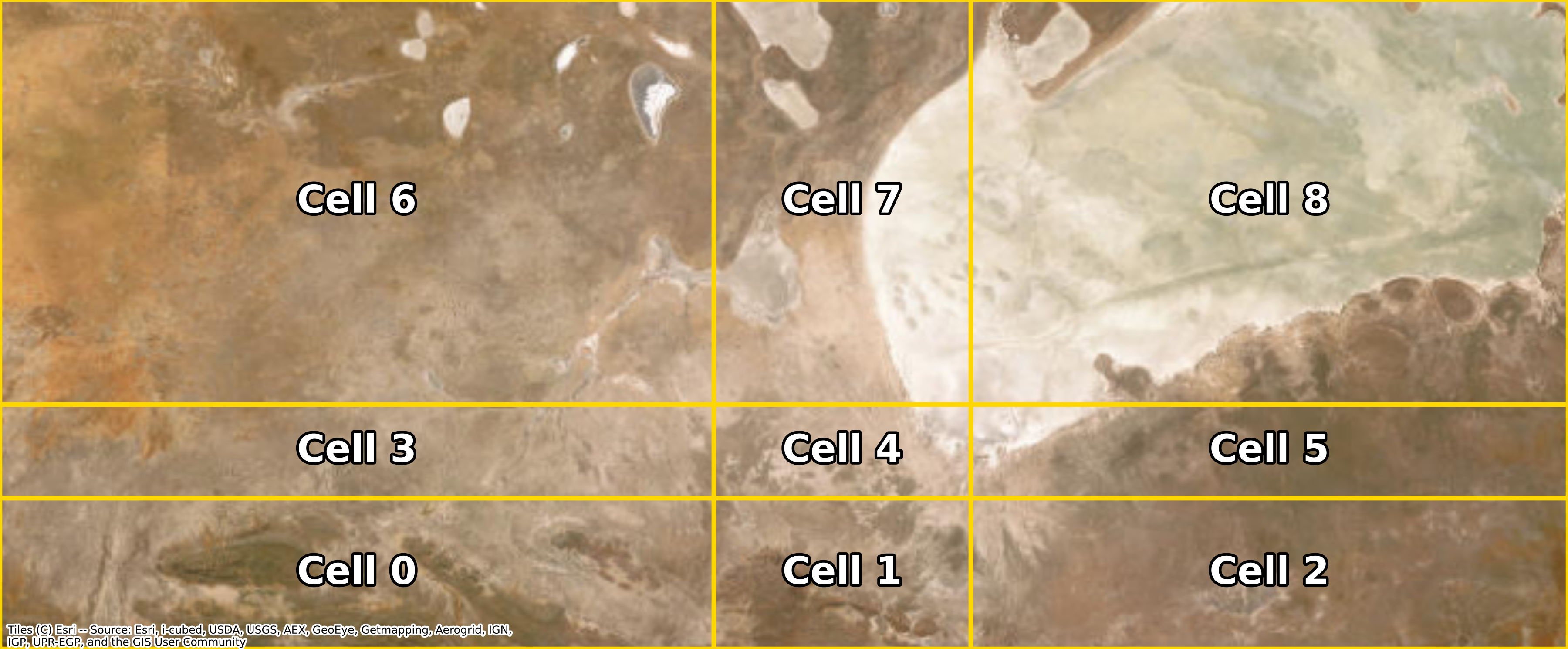}%
    }%
    \hfill
    \subcaptionbox{123 Manhattan subway stations used as targets \cite{mta2025ridership,osm2024}.\label{fig:mta_map}}{%
        \includegraphics[width=0.48\textwidth]{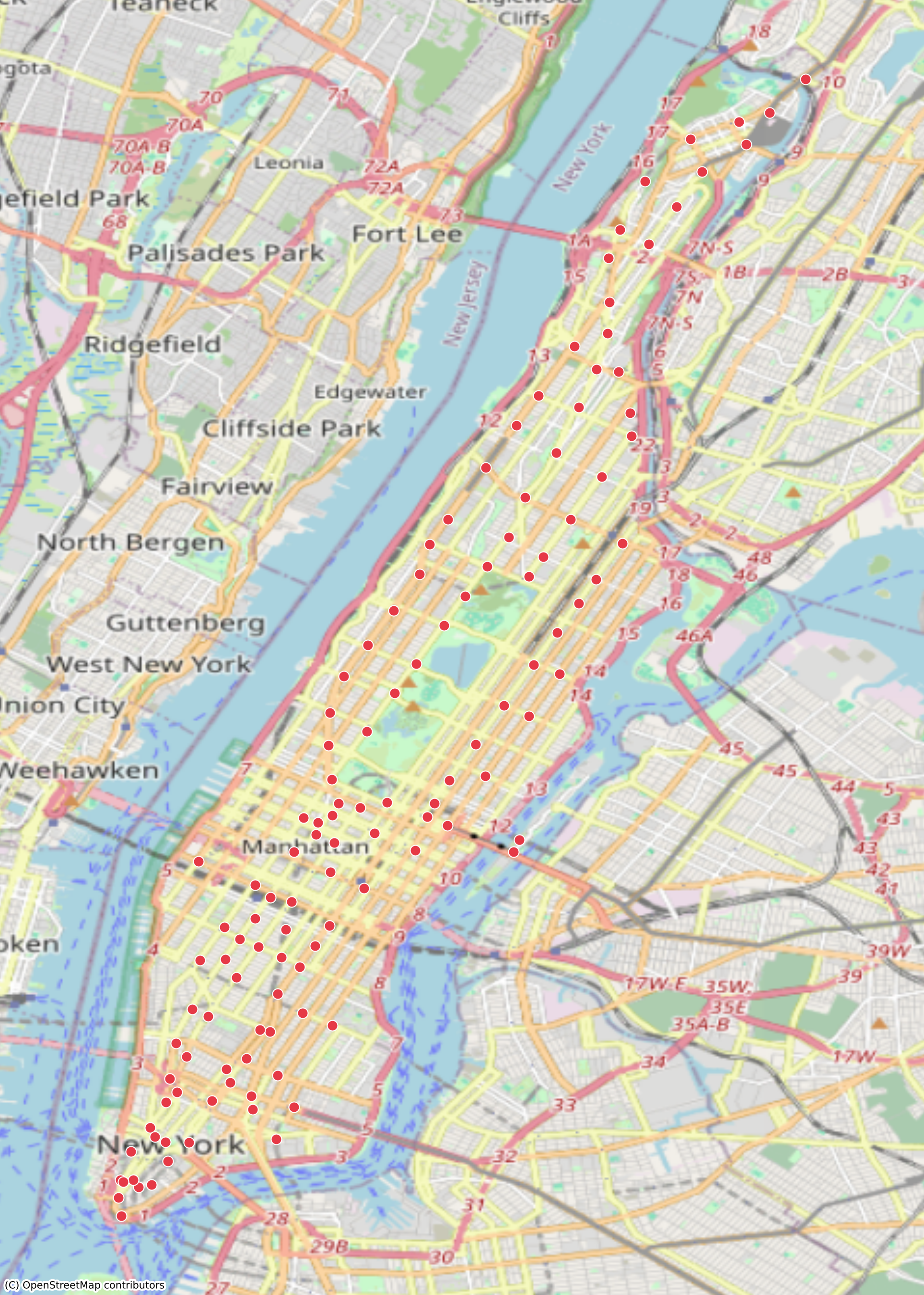}%
    }
    \caption{Geographic visualizations of the two game environments. (a) Etosha National Park with patrol-grid overlay (Conservation Security Game). (b) 123 Manhattan subway stations used as targets (Infrastructure Security Game).}
    \label{fig:env_maps}
\end{figure*}

\textbf{Conservation Security Game.} The cell values record elephant counts, drawn from a study tracking 15 African elephants within Etosha National Park published on the Movebank repository \cite{getz2018etosha}, covering movements from 2008 to 2014. The patrol-grid overlay partitions the reserve into nine cells whose boundaries are calibrated to elephant population density and its temporal variation: regions of high and variable activity receive finer-grained cells, while persistently low-activity regions are consolidated into coarser ones. This was to ensure that the experimental setup reflects where the distribution of patrol resources would be most meaningful. We construct 396 game instances spanning February 2013 to March 2014, one per evaluation day. We average the utilities across all instances for both ETE and PTO approaches and compare across training runs for models allowed to train up until a different number of epochs.

\textbf{Infrastructure Security Game.} The cell values record subway ridership, drawn from New York Metropolitan Transit Authority records \cite{mta2025ridership}, covering 123 Manhattan stations from 1 January 2022 to 31 December 2024. We construct 1098 game instances by averaging observations over 8-hour windows (approximately 3 instances per day). As with the conservation game, we average utilities across all instances for both approaches and compare across training epoch budgets.

\paragraph{Feature Engineering.}
For each observation, we extract temporal features capturing details about the period in which the observation occurs, e.g., the month, day of week. We additionally construct $W$ autoregressive lag features per cell, where the $i$-th lag corresponds to the value observed at that cell $i$ time periods prior. The full feature vector for each observation is thus $\mathbf{x} = (\text{temporal features}, v_{t-1}, v_{t-2}, \ldots, v_{t-W}) \in \mathbb{R}^{d}$, where $v_{t-i}$ denotes the $i$-th lagged value and $d$ is the total feature dimension. 

\paragraph{Train-Validation-Test Split.}
We employ a strict temporal split to prevent data leakage. For the infrastructure game, all observations from 2024 are held out as the test set, while the pre-2024 data is partitioned into training and validation sets. For the conservation game, we analogously reserve data from February 24 2013 up until March 28 2014 for final temporal segment for testing. In both cases, we identify the set of unique time periods in the training pool and randomly assign 15\% of these time periods---rather than individual rows---to the validation set. This ensures that all cells observed at a given time period fall entirely within either the training or validation fold, preventing the model from exploiting contemporaneous cross-cell information during validation.

\paragraph{Architecture.}
The Base MLP consists of two fully connected hidden layers. The input is projected to a hidden dimension of $h = 256$, followed by Layer Normalization, a SiLU activation, and dropout with probability $p = 0.2$. The second hidden layer reduces the dimension to $h/2 = 128$ with a SiLU activation, and a final linear layer projects to a single scalar output. A ReLU activation is applied at the output to enforce non-negativity, reflecting the physical constraint that state values cannot be negative. This architecture is shared across both domains. For the Conservation Security Game, the input encodes spatial cell features; for the Infrastructure Security Game, the input consists of 132 features per station — six lagged 8-hour ridership windows and a 123-dimensional station one-hot encoding — and the model is applied independently per station to predict the next ridership value.

\paragraph{Training Procedure.}
The supervised loss is the Huber loss, which is robust to outlier values compared with mean squared error. Optimization uses AdamW with weight decay $0.01$ and a ReduceLROnPlateau schedule that halves the learning rate after 5 epochs without improvement in validation Huber loss. The initial learning rate is $\eta = 10^{-3}$ (Conservation) or $\eta = 0.01$ (Infrastructure) for PTO and ZOC, and $\eta = 5 \times 10^{-4}$ for KKT and Hybrid in both domains. Gaussian input noise $\epsilon \sim \mathcal{N}(0, \sigma^2 I)$ with $\sigma = 0.01$ is applied during the supervised updates of PTO, ZOC and Hybrid; KKT uses no input noise, and KKT and Hybrid clip gradient norms at $1.0$. ZOC and Hybrid use regimes of $R = 50$ epochs. All models are trained for 500 epochs over 5 independent seeds. For the Infrastructure Security Game, the KKT and Hybrid models anneal the Huber regularization weight from $\lambda_{\max} = 5.0$ to $\lambda_{\min} = 3.0$ linearly over 300 epochs (for Hybrid, counted from the end of its first supervised regime), with strategic gradient steps restricted to a top-$k = 20$ station subset and up to 32 time windows sampled per epoch. In the Conservation Security Game we do not use Huber regularization in the strategic loss ($\lambda_{\max} = \lambda_{\min} = 0$), and up to 128 time windows are sampled per epoch.

% The model is trained using the Huber loss, which provides robustness to outlier values compared to standard mean squared error. Optimization is performed with AdamW using an initial learning rate of $\eta = 0.01$ and weight decay of $0.01$, with a ReduceLROnPlateau schedule that halves the learning rate after 5 epochs of no improvement in validation loss. Gaussian noise augmentation $\epsilon \sim \mathcal{N}(0, \sigma^2 I)$ with $\sigma = 0.01$ is applied to inputs during training. Each configuration is run over 5 independent seeds. The above applies to both domains. For the Infrastructure Security Game, the KKT and Hybrid models additionally anneal the Huber regularization weight from $\lambda_{\max} = 5.0$ to $\lambda_{\min} = 3.0$ linearly over the first 300 KKT epochs, with strategic gradient steps restricted to a top-$k = 20$ station subset and up to 32 time windows sampled per epoch. In the conservation security game we do not make use of Huber regularization by setting $\lambda_{max} = \lambda_{min} = 0$.

\paragraph{Evaluation.}
We evaluate along both statistical and strategic axes. On the statistical side, we report \emph{Mean Absolute Error} (MAE) between predicted and actual values across all cell-time pairs in the test set. On the strategic side, the predicted state vector $\hat{s}$ is used to construct the payoff matrices of the security game, from which a Strong Stackelberg Equilibrium (SSE) defender strategy is computed via linear programming and evaluated against the \emph{actual} state $s$ to obtain the realized defender utility. This applies to both domains. For the Infrastructure Security Game, the test set comprises 1,098 8-hour windows across 123 stations; we report \emph{per-window utility} (total utility $/ 1098$) as the primary metric.

\textbf{Compute.} All experiments were run on a single GPU; full runs (5 seeds, 500 epochs) take approximately 1--2 hours per domain for PTO and ZOC. For the Infrastructure Security Game, KKT and Hybrid runs take approximately 3--4 hours due to per-epoch LP gradient calls via \texttt{cvxpylayers}. CPU-only execution is supported but substantially slower for the differentiable LP methods.

\subsection*{Controlled Experiments}

% ── Required packages ──────────────────────────────────────────────────────
% \usepackage{booktabs, threeparttable, graphicx, subcaption, amsmath}
% ───────────────────────────────────────────────────────────────────────────

\subsection{Synthetic Conservation Security Game}\label{sec:app-toy-exp}

\paragraph{Setup.}
We design a synthetic security game ($N=100$ cells, catch
penalty $P=10$) to explore \emph{when} each training objective succeeds or
fails, rather than to produce aggregate rankings.
The exogenous state is an animal density map (non-negative, peak cell
normalised to $10$).
The defender receives $n_{\mathrm{obs}}=5$ independent sparse noisy sweeps,
each visiting $20$ of $100$ cells uniformly at random with Poisson-corrupted
counts.
All four methods share the same four-layer CNN (convolutional with instance
normalisation, SiLU activations, Softplus output) trained with AdamW on
$1{,}000$ synthetic scenarios for $200$ epochs.

\textbf{Training objectives.}
\textbf{PTO} minimises Huber loss on predicted density and retains the
final epoch.
\textbf{ZOC} (Zero-Order Checkpointing) trains identically but evaluates
downstream strategic utility on a held-out $15\%$ of training data every
$20$ epochs, retaining the checkpoint that maximises that score.
\textbf{KKT} directly maximises defender utility by backpropagating through
the Stackelberg LP via \texttt{cvxpylayers}: the attacker's best-response
target $t^*$ is found by binary search (detached from the graph), the LP is
re-solved for $t^*$ with gradients on, and exact KKT sensitivities flow back
to model parameters.
\textbf{Hybrid} alternates between $20$-epoch Huber regimes and $20$-epoch
KKT regimes, reverting to the running-best checkpoint after each regime if
strategic utility does not improve.

Density states are drawn from three families:
(i)~\emph{Gaussian blobs} — $2$–$4$ randomly placed modes
($\sigma\in[0.05,0.25]$);
(ii)~\emph{sparse spikes} — $1$–$6$ isolated cells with non-zero value;
and (iii)~\emph{concentric rings}.
Crucially, the three families induce structurally different attacker
best-response landscapes, letting us study regime-dependent behaviour
rather than a single aggregate score. We present the qualitative observations below:

\paragraph{When does Hybrid help?}
Hybrid is most effective on \emph{spatially extended multi-modal} densities
(Figure~\ref{fig:regime-main}, top row), where it outperforms both ZOC and
KKT alone by up to $0.8$ utility units. The revert-to-best filter
consistently accepts the KKT phase on these instances.

\paragraph{When does KKT alone win?}
On \emph{sparse spike} densities (Figure~\ref{fig:regime-main}, second row),
KKT alone achieves the highest utility. Hybrid's revert-to-best filter which leads to 
rejection of some KKT updates, results in sub-optimality for this particular setting.

\paragraph{When does PTO or ZOC win?}
On densities with \emph{two well-separated modes of unequal weight}
(Figure~\ref{fig:regime-main}, third row), accuracy ends up being a very important part of determining the most dense mode. This plays an outsized role in shaping utility.

\paragraph{Summary.}
Hybrid is the natural default for multi-modal, spatially extended densities.
KKT alone dominates when the density is sparse and the best-response target
is unambiguous; it can fail when the landscape is smooth or multi-modal.
PTO and ZOC remain competitive when predictive accuracy and strategic
utility are well-aligned, most commonly for single-dominant-mode
distributions. Designing a switching criterion that conditions Hybrid's
acceptance step on the \emph{stability of $t^*$} rather than a scalar
utility threshold is a key direction for future work.

\begin{figure}[t]
  \centering
  \includegraphics[width=\linewidth]{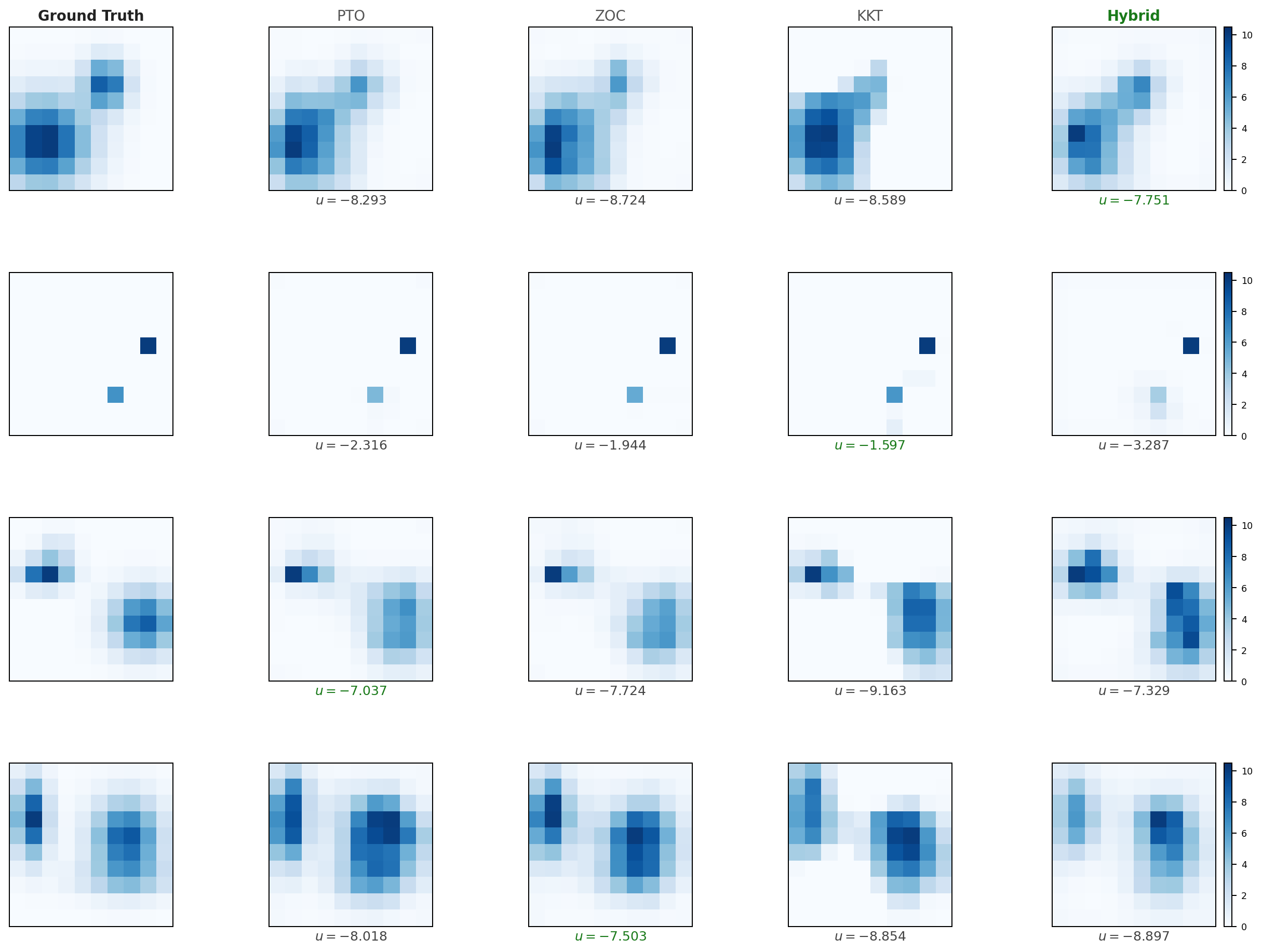}
  \caption{%
    Four qualitative regimes on the synthetic $10\times10$ security game.
    Each row is one test scenario; columns show the ground-truth density
    and the prediction of each method.
    The winning method (highest defender utility against the true density)
    is labelled in \textcolor{green!50!black}{\textbf{bold green}};
    utility $u$ is reported below each prediction.
    \textbf{Row 1} (\emph{multi-modal blob}): \textbf{Hybrid} wins.
    \textbf{Row 2} (\emph{sparse spike}): \textbf{KKT} wins .
    \textbf{Row 3} (\emph{two separated modes}): \textbf{PTO} wins.
    \textbf{Row 4} (\emph{single blob with secondary cluster}): \textbf{ZOC}
    wins.%
  }
  \label{fig:regime-main}
\end{figure}

\begin{figure}[t]
  \centering
  \begin{subfigure}[b]{\linewidth}
    \includegraphics[width=\linewidth]{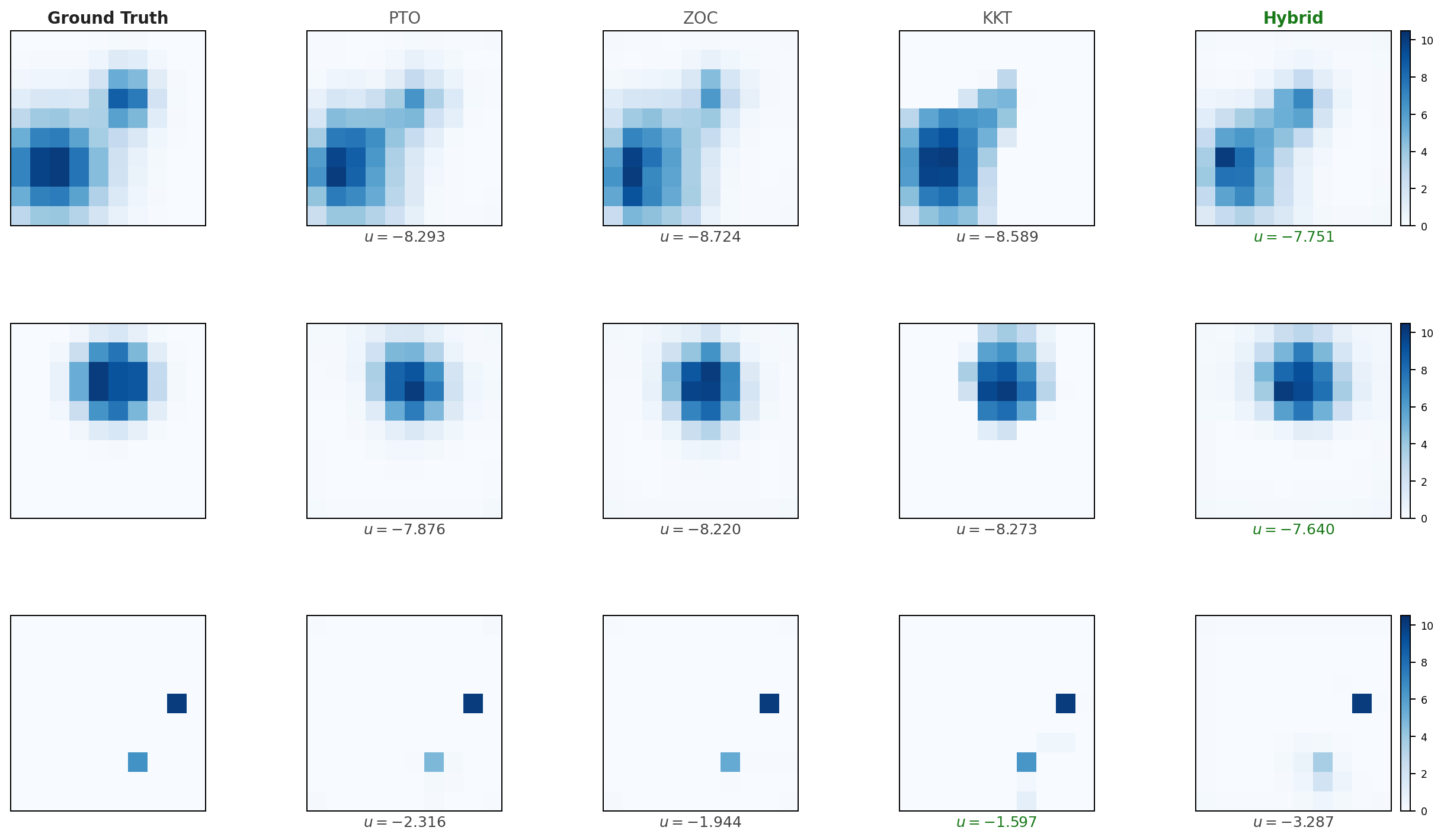}
    \caption{%
      Two scenarios where \textbf{Hybrid} wins (rows 1--2) alongside a
      scenario where \textbf{KKT} wins (row 3).%
    }
    \label{fig:hybrid-kkt}
  \end{subfigure}

  \vspace{1em}

  \begin{subfigure}[b]{\linewidth}
    \includegraphics[width=\linewidth]{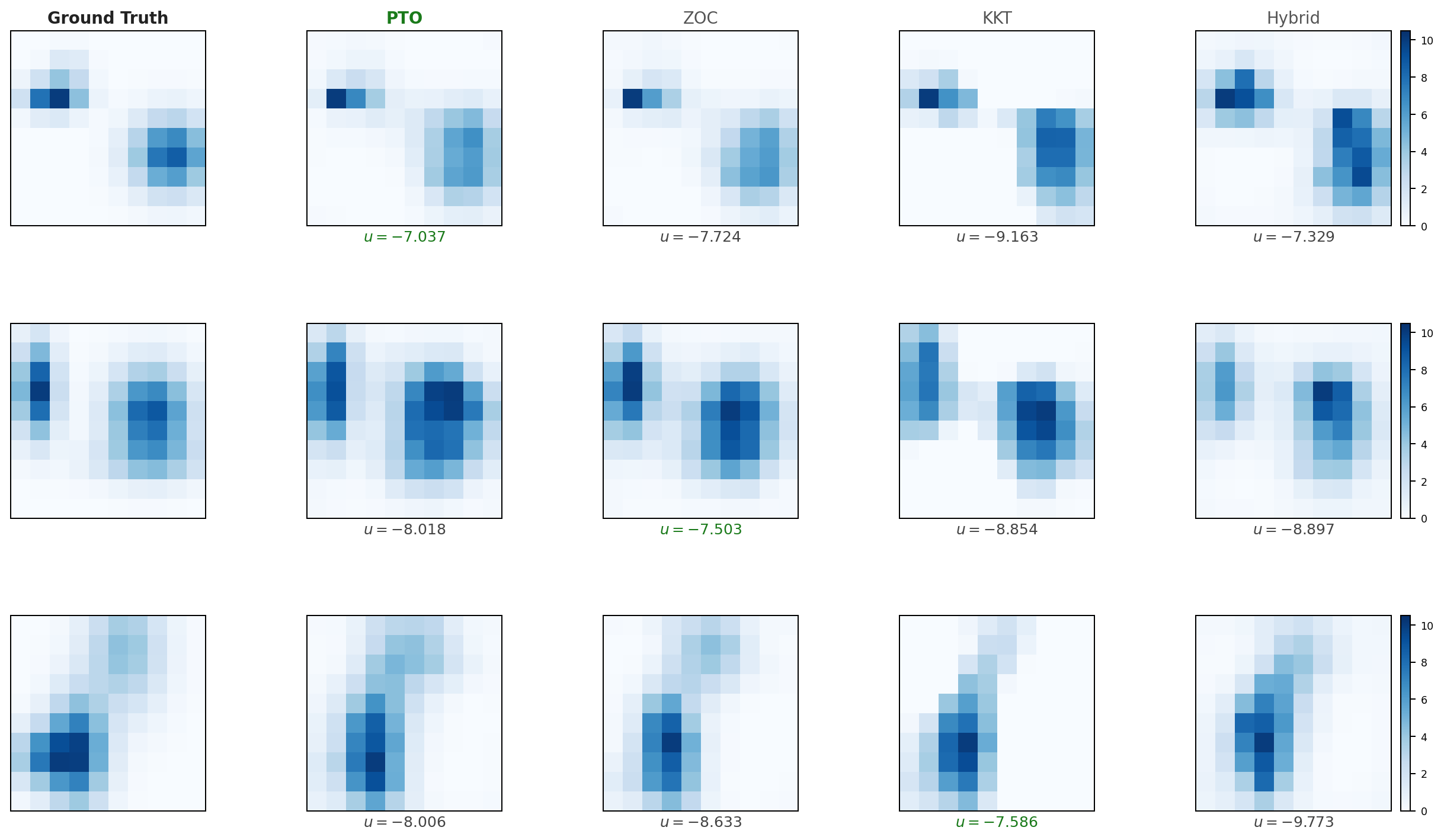}
    \caption{%
      \textbf{Row 1}: PTO wins.
      \textbf{Row 2}: ZOC wins.
      \textbf{Row 3}: KKT wins.%
    }
    \label{fig:failures}
  \end{subfigure}
  \caption{%
    Scenario-level breakdown illustrating the conditions under which
    Hybrid, KKT, PTO, and ZOC each succeed or fail.
    Aggregate metrics obscure these structural differences; the regime
    panels above make them explicit.%
  }
  \label{fig:breakdown}
\end{figure}
\FloatBarrier
\newpage

\end{document}